\documentclass[a4paper,10pt]{scrartcl}

\usepackage{actuarialsymbol}
\usepackage{amsmath}
\usepackage{amssymb}
\usepackage{amsthm}
\usepackage{bbm}
\usepackage{float}
\usepackage{chngcntr}
\usepackage[margin=3cm]{geometry}
\usepackage{graphicx}
\usepackage{lmodern}
\usepackage{natbib}
\usepackage[normalem]{ulem}
\usepackage[table]{xcolor}
\usepackage{caption, subcaption}
\usepackage{url}

\newtheorem{lem}{Lemma}[section]

\newtheorem{claim}[lem]{Claim}
\newtheorem{thm}[lem]{Theorem}
\newtheorem{cor}[lem]{Corollary}
\newtheoremstyle{NoItalic}{}{\baselineskip}{}{}{\bfseries}{.}{.5em}{\thmname{#1}\thmnumber{ #2}\thmnote{ #3}}\theoremstyle{NoItalic}
\setcapindent{0pt}

\begin{document} 

\title{Quantifying the trade-off between income stability and the number of members in a pooled annuity fund}
\author{Thomas Bernhardt\footnote{ {\texttt{T.Bernhardt@hw.ac.uk}}. Department of Mathematics, University of Michigan, Ann Arbor, Michigan MI 48109-1043, USA.} \hspace{0em} and Catherine Donnelly\footnote{{\texttt{C.Donnelly@hw.ac.uk}} (corresponding author). Risk Insight Lab, Department of Actuarial Mathematics and Statistics, Heriot-Watt University, Edinburgh, Scotland EH14 4AS.}}

\date{\today}
\maketitle

\vspace{-4em}

\section*{\begin{center}Abstract\end{center}}
\begin{abstract}
\noindent The number of people who receive a stable income for life from a closed pooled annuity fund is studied.  Income stability is defined as keeping the income within a specified tolerance of the initial income in a fixed proportion of future scenarios.  The focus is on quantifying the effect of the number of members, which drives the level of idiosyncratic longevity risk in the fund, on the income stability.  To do this, investment returns are held constant and systematic longevity risk is omitted.   An analytical expression that closely approximates the number of fund members who receive a stable income is derived and is seen to be independent of the mortality model.  An application of the result is to calculate the length of time for which the pooled annuity fund can provide the desired level of income stability.

\textbf{Keywords:} Longevity credit; decumulation; tontine; unsystematic risk; pooling.
\end{abstract}

\section{Introduction}\label{section:Intro}
The lack of innovative products to convert a lump-sum into a life-long retirement income is increasingly seen as a problem in the United Kingdom \citep[Paragraph 12, page 10]{frankfield2018}.  Previously, pension savers typically either bought a life annuity from an insurance company or saved within a defined benefit pension scheme.  With the former option rather unattractive for those with smaller pension pots due to low annuity rates (caused by low investment returns, uncertainty about future lifespans and solvency capital requirements) and the latter becoming less available due to the closing of defined benefit pension schemes, there is an opportunity for attractive, alternative ways to convert savings into an income for life.

Pooled annuity funds, a type of tontine, are an alternative which offer the opportunity of a lifelong, stable income to its participants.  They pool the retirement wealth of their participants and pay it out to the survivors as a regular income.  They should give a higher income than under self-insurance (also called income drawdown in UK terminology), for the same chance of exhausting funds.  They can do this because they pool the participants' longevity risk, and hence subsidize the payments to the longer-lived members, and they do not pay a bequest to the participants' beneficiaries.  It is the same principle underpinning life annuities.  The critical difference is that the income from a conventional life annuity product is guaranteed by an insurance company whereas that from a pooled annuity fund is not guaranteed by anyone.  The income from a pooled annuity fund will vary due to fluctuations in investment returns and longevity experience, unlike those of a conventional life annuity, which gives a constant income regardless of investment and longevity experience.

In this paper a mathematical expression is derived which allows the calculation of a lower bound on the number of participants who receive a stable, life-long income once the degree of income volatility is specified.  A numerical study suggests that the lower bound is close in value to the true value.  By keeping investment returns constant and not including systematic longevity risk (the risk that the wrong distribution of future lifetime has been chosen), the risk attributable to idiosyncratic longevity risk (the risk that the distribution of future lifetime is not observed perfectly among the annuitants) is studied in isolation.  This is important since the \emph{raison d'\^etre} of joining the fund for individual participants - a stable income paid for life - is reliant on idiosyncratic longevity risk being sufficiently diversified.  The analysis in this paper will allow the risk inherent in random investment returns and systematic longevity risk to be quantified and attributed separately in further research.  Importantly, the derived results give us the insight that the proportion of people who receive a stable, life-long income is approximately independent of the mortality model, once the number of members initially in the fund is fixed.

Tontines have enjoyed recently some attention in the academic literature.  \citet{Piggottetal2005} proposed Group Self-Annuitization, which is the structure studied here.  In it, the income paid to the annuitants is calculated from their individual fund values and annuity factors.  The income calculated for each annuitant should be payable for life as long as both the investment and mortality experience turn out as assumed in the annuity factors.  Group Self-Annuitization is similar to the Pooled Annuity Fund introduced in \citet{stamos2008}, although the latter does not specify how income is withdrawn from their structure, but lets income be withdrawn as the annuitant desires.  These structures are designed to operate when all annuitants are independent and identical copies of each other.  The same membership assumption is made in this paper; indeed, the pooled annuity fund studied here is that of \citet{Piggottetal2005} in the homogeneous case.  An extension of Group Self-Annuitization to heterogeneous groups is given in \citet{valdezetal2006}.

Alternative ways of sharing out the funds in a tontine are proposed by \citet{donnellyetal2014} and \citet{sabin2010}, with the latter further discussed in \citet{FormanSabin2015}.  Their tontines pay survivors an explicit longevity credit, which represents a share of the funds of the recently deceased, and the level of income withdrawn is not mandated.  In contrast, tontines like those of \citet{MilevskySalisbury2015} use the level of income withdrawn to implicitly distribute the funds of the deceased among the survivors.   \citet{chenetal2019} proposes and analyzes a product that is comprised of a tontine until a fixed age and a deferred life annuity that comes into payment at the fixed age.  \citet{DonnellyYoung2017} propose a tontine product with a minimum guaranteed payment, and the proposed product is developed and examined further in \citet{chenrach2019}.

The question of actuarial fairness - that the expected value of the benefits received should equal the value of each annuitant's fund - can arise in a discussion of tontine structures.  Some are actuarially fair only for homogeneous memberships, such as the Group Self-Annuitization and the Pooled Annuity Fund (as shown in \citealt{Donnelly2015}).  Others are actuarially fair for heterogeneous memberships, such as those in \citet{donnellyetal2014} and \citet{sabin2010}.  \citet{MilevskySalisbury2016} argue instead for actuarial equitability rather than fairness - everyone gets the same expected value of benefit payments even if it is less than what they put in - and propose a structure to do this.

\citet{qiaosherris2013} analyze and study how to manage systematic longevity risk in Group Self-Annuitization structures.  Noting that annuity payments decline at older ages if the annuity factors used to calculate the income withdrawn are based on up-to-date mortality experience, they suggest a way of sharing systematic longevity risk across cohorts, through an adjustment to the calculation of \citet{Piggottetal2005}.  They do a numerical study of this, calculating the quantiles of the income distribution at fixed time points.

However, our aim is to study idiosyncratic longevity risk.  The mathematical results presented hold for any mortality model as long as systematic longevity risk is excluded.  The evolution of each sample path of the income is analysed, which permits a precise quantification of the behaviour of the income streams.  \citet{Piggottetal2005} and \citet{sabin2010} both consider a particular mortality model and plot the income amounts, observing that they are more volatile at older ages.  Their conclusions are qualitative and not quantitative.  None of these studies defines income stability and so cannot study it rigorously and generally, as is done here.  While one definition of income stability is proposed in this paper, there are many other possible definitions which may be more suitable for different needs. 

After the operation of the pooled annuity fund is detailed in Section \ref{section:operation}, the results in Section \ref{subsection:main-theorem} enable the calculation of the number of people who can receive a life-long, stable income from the pooled annuity fund.  In effect, this is the number of people who die while the desired level of income is stable.  Our results give a lower bound which can be used as an approximation (Section \ref{subsection:efficacy}) to the true number of people.  Importantly, the approximation is independent of the mortality model and thus the numbers calculated hold for any age, sex or other sub-population, once the fund parameters are fixed.

Determining for how long a pooled annuity fund can pay a stable income to its members from the theoretical results is studied in Section \ref{subsection:confidence}.  Finally, an approximation to the lower bound is presented in Section \ref{section:Donsker-approach}, with the derivation relegated to the Appendix.

\section{The operation of the pooled annuity fund}\label{section:operation}

Consider a group of $N \geq 2$ individuals who constitute the entire membership of the pooled annuity fund at time 0.  Each member is an independent and identical copy of the rest.  No member joins after time 0 and no member can leave except through death.

\subsection{Future lifetime random variables and survival probabilities}

The members are each aged $x\geq 0$ at time 0.  Represent the $i$th member's future lifetime from age $x$ by the real-valued random variable $T_{i} > 0$, for $i=1,\ldots,N$.  The random variables $T_{1}, T_{2}, \ldots, T_{N}$ are defined on the probability space $(\Omega,\mathcal{F},\mathbb{P})$, are independent and identically distributed and have a continuous distribution.  For each $\omega \in \Omega$, the order statistics $(T_{(i)}(\omega))_{i=1}^N$ are the increasingly ordered $(T_i(\omega))_{i=1}^N$, i.e. $T_{(1)} \leq T_{(2)} \leq \cdots \leq T_{(N)}$.   

The number of individuals observed to be alive at age $x+t$ is
\begin{equation*}
    L_{x+t}=\sum_{i=1}^{N} \mathbbm{1}_{[T_i>t]}, \quad\mbox{for $t\geq0$},
\end{equation*}
in which $\mathbbm{1}_{A}$ is the zero-one indicator function of the set $A \subset \Omega$.

The empirical survival probability to age $x+t+s$ conditional on being alive at age $x+t$, for $s,t \geq 0$, is
\[
{}_{s} \hat{p}_{x+t} = L_{x+t+s} / L_{x+t},\quad\mbox{if $L_{x+t} > 0$},
\]
and the assumed true survival probability from age $x+t$ to age $x+t+s$ is
\[
{}_{s} p_{x+t}=\mathbb{P} \left[ T_i>t+s \, \big\lvert \,T_i>t \right].
\]
In line with actuarial convention, $\hat{p}_{x+t}:={}_{1}\hat{p}_{x+t}$ and $p_{x+t}:={}_{1}p_{x+t}$.  The calculation of longevity credits paid to survivors in the fund uses the empirical survival probabilities $({}_{s} \hat{p}_{x+t})_{s,t \geq 0}$, whereas the calculation of the income withdrawn from the members' account values relies on the assumed true survival probabilities $({}_{s} p_{x+t})_{s,t \geq 0}$.

\subsection{The income calculation} \label{SUBSECincomeprocess}
Each member of a pooled annuity fund has a fund account with value $W(t)\geq0$ at time $t$, with constant $W(0)>0$ at time $0$.  The fund account value changes over time due to investment returns, income withdrawals and longevity credits, the latter coming from the re-allocation of the fund accounts of the members who died. 

The account values are invested to get a constant effective rate of return $R > -1$ over each unit of time.

The income withdrawn from each account by each surviving member is the amount that a fair life annuity would pay if purchased with the current fund value.  Thus the member withdraws an income of amount
\begin{equation} \label{eqn:incomecalculation}
    C(t) = W(t) / \ddot{a}_{x+t} \quad\mbox{at $t=0,1,2,\ldots$},
\end{equation}
in which 
\[
\ddot{a}_{x+t} = 1 + \sum_{j=1}^{\infty} (1+R)^{-j} \, {}_{j}p_{x+t}.
\]

\subsection{The longevity credit calculation}

The account value at time $t+1$ of a member who dies over the time period $(t,t+1]$ is $(W(t)-C(t))(1+R)$.  The number of deaths observed over $(t,t+1]$ is $L_{x+t}-L_{x+t+1}$ and so an amount equal to $(W(t)-C(t))(1+R)(L_{x+t}-L_{x+t+1})$ is distributed equally among the $L_{x+t+1}$ survivors observed at time $t+1$.  Thus the longevity credit paid at time $t+1$ to each member alive at time $t+1$ is
\begin{equation} \label{eqn:longevitycredit}
    M(t+1)=(W(t)-C(t))(1+R)(L_{x+t}-L_{x+t+1}) / L_{x+t+1},\quad\mbox{if $L_{x+t+1} > 0$}.
\end{equation} 
If $L_{x+t+1}=0$ then there is no-one left alive at time $t+1$, and the pooled annuity fund ceases to exist.  With no-one left in the fund at time $t+1$, the account values of those who died over $(t,t+1]$ would be paid to their estate at time $t+1$.

Immediately after the payment of the longevity credit, the account value at time $t+1$ of a member who is alive at that time is
\begin{equation} \label{eqn:wealthwithM}
    W(t+1)=(W(t)-C(t))(1+R) + M(t+1),\quad\mbox{for $t\geq0$}.
\end{equation}
In contrast, the account value of a member who is dead at time $t+1$ is zero since the funds of deceased members are distributed among the survivors. 

Using the identity $\ddot{a}_{x+t}-1=p_{x+t} \, \ddot{a}_{x+t+1}/(1+R)$ and substituting for $M(t+1)$ from equation (\ref{eqn:longevitycredit}) and for $W(t)$ and $W(t+1)$ from equation (\ref{eqn:incomecalculation}) into equation (\ref{eqn:wealthwithM}), shows that the income at time $t+1$ can be written as
\begin{equation} \label{eqn:incomewithoutM}
C(t+1) = C(t)\;p_{x+t} / \hat{p}_{x+t},\quad\mbox{if $L_{x+t+1} > 0$}.
\end{equation}
Thus fluctuations in the income are caused only by fluctuations in the observed survival probability $\hat{p}_{x+t}$ against the true survival probability $p_{x+t}$ and the level of investment returns do not affect the fluctuations.  This is a consequence of setting the investment returns to be constant and known, which allows idiosyncratic mortality risk to be studied in isolation.   The income in one future scenario is shown in Figure \ref{fig:fluctuating-income}.

\begin{center}
\includegraphics[clip,trim={0 0.5cm 0 1.25cm},width=0.6\linewidth]{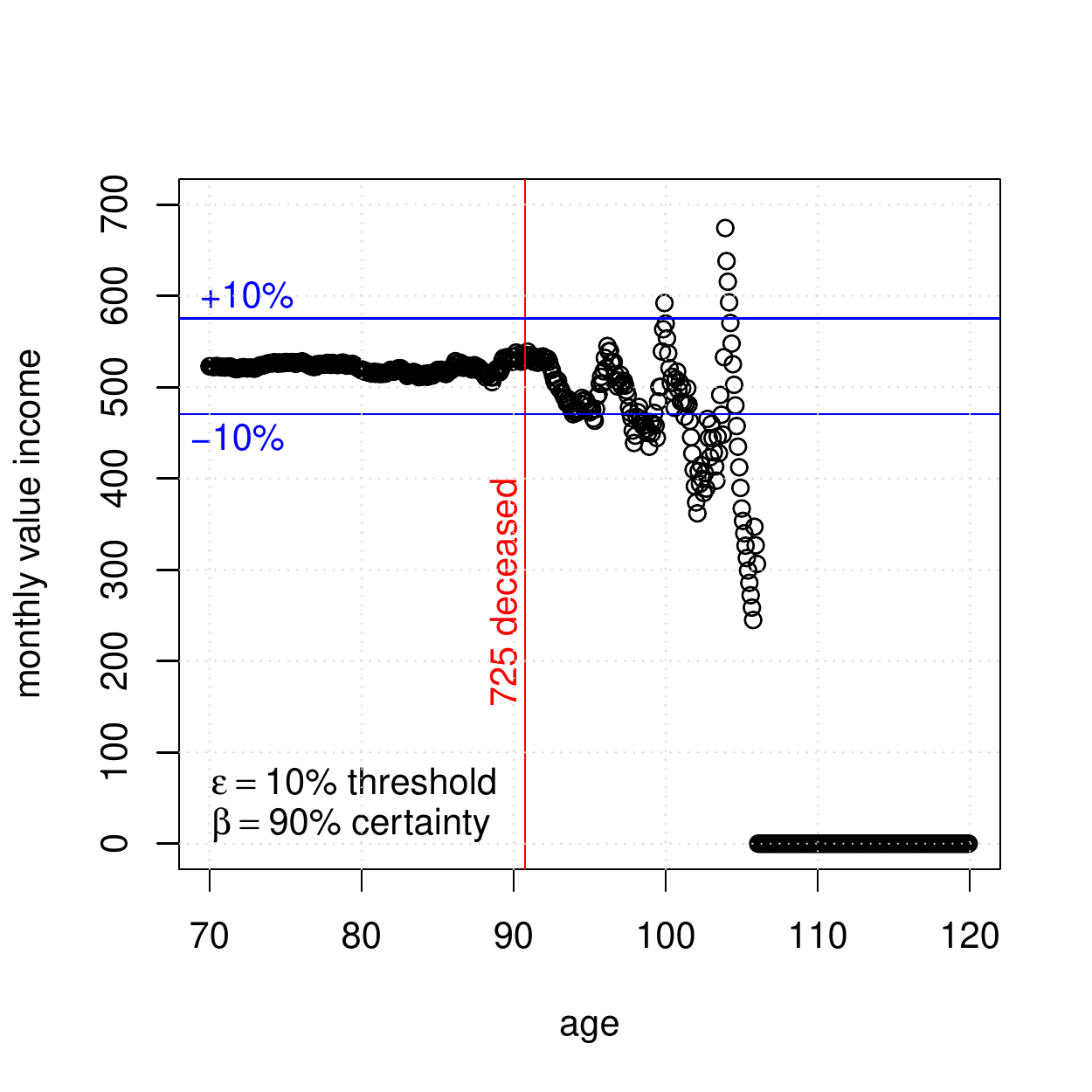}
\captionof{figure}{
    A sample of the income process, shown as black circles, for a fund with initially $N=1000$ members, each aged $x=70$ and bringing $W(0)=100\,000$ units of wealth to the fund at time $0$.  Time is measured in months, the mortality model used is the UK-based life table S1PFL \citep{IFoA2008} and a uniform distribution of deaths is assumed between integer ages.  The blue horizontal lines indicate the lower and upper income thresholds ($\varepsilon_{1}=\varepsilon_{2}=0.1$).  The plot shows only one sample.  However, for $90\%$ of all samples, the income stays between the income thresholds until $725$ members died.  In the sample income process shown, the income happens to stay within the band for a few years after the time of the $725$th death.
}
\label{fig:fluctuating-income}
\end{center}
If idiosyncratic longevity risk is completely diversified in the pooled annuity fund then, for all $t \geq 0$, ${}_{t} \hat{p}_{x} = {}_{t} p_{x}$ and it would follow that $C(t)=C(0)$.  However, the analysis in this paper focuses on less-than-perfect pooling of idiosyncratic longevity risk.  The key investigation is to measure how well does the income stream $(C(t); t=0,1,2,\ldots)$ stay close to the initial amount of income withdrawn, $C(0)$.

\section{The problem}\label{section:problem}
The purpose of the pooled annuity fund is to pay a regular, stable income to the fund's members for as long as they live.  The level of the initial income withdrawn, $C(0)$, is important since the stability of the subsequent income withdrawn by the annuitants is measured by reference to it.

Fix a value $\varepsilon_{1}\in(0,1)$, called the lower threshold parameter, a value $\varepsilon_{2} >0$, called the upper threshold parameter and a value $\beta\in(0,1)$, called the certainty.  The income $C(\omega, t)$ at time $t>0$ and in future scenario $\omega$ is deemed to be stable if $C(\omega, t) \in [(1-\varepsilon_{1})C(0), (1+\varepsilon_{2})C(0)]$.  The values $(1-\varepsilon_{1})C(0)$ and $(1+\varepsilon_{2})C(0)$ are called the lower and upper income thresholds, respectively.  The aim is to measure for how long is the income stable in at least $100\beta$\% of the future scenarios.  Denote the number of people who receive an income in the range $[(1-\varepsilon_{1})C(0), (1+\varepsilon_{2})C(0)]$ in at least $100\beta$\% of the future scenarios for the whole of their lifetime by $k_{C}$.  Then all participants receive a stable income with certainty $\beta$ up to the time of the $k_{C}$th death.  But it is only the $k_{C}$ annuitants who die first who receive a stable income with certainty $\beta$ for the whole of their lifetime, rather than for the first part of their lifetime only.

The motivation for the definition of income stability is that it should be understandable to finance professionals.  The closest definition to ours is the ``4\% rule'', introduced in \citet{Bengen2001} and further studied in \citet{Guyton2004}, \citet{GuytonKlinger2006} and \citet{PfauKitces2014}.  The idea is to calculate the initial amount of income that can be withdrawn for 30 years, with subsequent year's withdrawals adjusted for inflation, such that there is at least a 90\% chance of this strategy being sustainable.  In contrast to the probabilistic approach taken here, an alternative approach can be to maximize the expected value of a function of the income withdrawn.  For example, \citet{HeLiang2013} minimize the discounted expected value of the squared distance of the income from a target value.  Others such as \citet{BruhnSteffensen2013}, \citet{Constantinides1990}, \citet{Munk2008} and \citet{vanBilsenetal2017} are based on habit formation.  The latter paper and \citet{Curatola2017} assume a loss averse investor.

The time when exactly $k$ members have passed away is denoted by the random variable $T_{(k)}$, the $k$th order statistic of the future lifetime random variables.   Suppose that the maximum integer $k:=k_{C}$ to satisfy
\begin{equation}\label{eq:want-T}
    \mathbb{P} \left[ (1+\varepsilon_{2}) C(0) \geq C(s) \geq (1-\varepsilon_{1})C(0) \;\;\mbox{for all $s \in \{1,2,\ldots, \lfloor T_{(k)} \rfloor \}$} \right] \geq \beta,
\end{equation}
in which $\lfloor T_{(k)} \rfloor$ is the integer part of $T_{(k)}$, is to be determined.  Figure \ref{fig:fluctuating-income} illustrates a choice of the income thresholds, with a realisation of the income stream.

In this paper, a close, lower bound to $k_{C}$ is found, which holds for any mortality model that excludes systematic longevity risk.  This surprising independence result is a consequence of the inverse transform method, which is applied in Theorem \ref{theorem:main} below.  It means that the values of the lower bound calculated for a choice of $\varepsilon$, $\beta$ and $N$ remain the same regardless of what distribution is chosen for $T_1, T_2, \ldots, T_N$.  In contrast, calculating the value of $k_{C}$ directly from (\ref{eq:want-T}) does require the choice of a mortality model.

It would be ideal to find an explicit expression for the maximum integer that satisfies (\ref{eq:want-T}), rather than a lower bound on it.  Indeed, the exact distribution of the order statistics of $T_1, T_2, \ldots, T_N$ are known \citep{Birnbaum1969, Miklos1965}.  Unfortunately, they are not amenable to calculations since their distribution functions are polynomials of degree $N$, and in a pooled annuity fund the value of $N$ represents the initial number of members.  Thus the value of $N$ is likely to be in the order of hundreds or more.  

Going further, the preferred goal would be to find the maximal integer time $t \geq 1$ such that
\[
\mathbb{P} \left[ (1+\varepsilon_{2}) C(0) \geq C(s) \geq (1-\varepsilon_{1}) C(0)\;\;\mbox{for all $s \in \{1,2,\ldots,t\}$} \right] \geq\beta.
\]
However, the maximal integer time would depend on a mortality model and would not yield a general result, like the one in this paper.

\subsection{The main results}\label{subsection:main-theorem}
\begin{thm} \label{theorem:main}
Let $U_{(1)}, U_{(2)}, \ldots, U_{(N)}$ be the order statistics of $N$ independent and standard uniformly distributed random variables $U_{1}, U_{2}, \ldots, U_{N}$.  Similarly, denote the order statistics of the independent and identically distributed future lifetime random variables $(T_{i})_{i=1}^{N}$, which have a continuous distribution, by $T_{(1)}, T_{(2)}, \ldots, T_{(N)}$.  Fix constants $\varepsilon_1 \in (0,1)$, $\varepsilon_2>0$ and $k \in \{1,2,\ldots,N \}$.  Then
    \[
    \begin{split}
    & \mathbb{P} \left[ (1+\varepsilon_2) C(0) \geq C(s) \geq (1-\varepsilon_1) C(0)\;\;\mbox{for all $s \in \{1,2,\ldots, \lfloor T_{(k)} \rfloor \}$} \right] \\
    \geq & \mathbb{P} \left[(1-\varepsilon_1) \tfrac{i-1}{N} + \varepsilon_1 \geq U_{(i)} \geq (1+\varepsilon_2) \tfrac{\min\{i,N-1\}}{N} - \varepsilon_2\;\;\mbox{for all $i\in\{1,2,\ldots, k\}$}\right],
    \end{split}
    \]
		in which $C(s)$ is the income at time $s$ that is calculated via equation (\ref{eqn:incomecalculation}).
\end{thm}
\begin{proof}
Fix $t \in \mathbb{N}$ and assume that at least one person is alive at age $x+t+1$, i.e. $L_{x+t+1} > 0$.  Starting with equation (\ref{eqn:incomewithoutM}), it follows by induction on $t \in \mathbb{N}$ that
\[
 C(t+1)= C(0)\,\prod_{j=0}^{t}\frac{p_{x+j}}{\hat{p}_{x+j}} = C(0)\,\frac{{}_{t+1} p_{x}}{{}_{t+1}\hat{p}_{x}}.
\]
As the joint distribution of $T_{1}, T_{2}, \ldots, T_{N}$ is continuous, the set $[\lfloor T_{(k)} \rfloor<T_{(k)}, \, \textrm{for $k= 1,2,\ldots, N$}]$ has measure one.  Noting that ${}_{s}\hat{p}_{x}>0$ for $s \in \{1,2,\ldots, \lfloor T_{(k)} \rfloor \}$ on this set and working only on this set, it follows that
\[
\begin{split}
 & \left[ (1+\varepsilon_2) C(0) \geq C(s) \geq (1-\varepsilon_1) C(0)\;\;\mbox{for all $s \in \{1,2,\ldots, \lfloor T_{(k)} \rfloor \}$} \right] \\
= & \left[ 1+\varepsilon_2 \geq \frac{{}_{s}p_{x}}{{}_{s}\hat{p}_{x}} \geq 1-\varepsilon_1 \;\;\mbox{for all $s \in \{1,2,\ldots, \lfloor T_{(k)} \rfloor \}$} \right] \\
= & \left[ \inf_{s \in \{ 1,2,\ldots, \lfloor T_{(k)} \rfloor  \} } \frac{{}_{s}p_{x}}{{}_{s}\hat{p}_{x}} \geq 1-\varepsilon_1 \right] \cap \left[ \sup_{s \in \{ 1,2,\ldots, \lfloor T_{(k)} \rfloor  \} } \frac{{}_{s}p_{x}}{{}_{s}\hat{p}_{x}} \leq 1+\varepsilon_2 \right].
\end{split}
\]

Let both $T_{(N+1)}:=\infty$ and ${}_{s}p_{x}/{}_{s}\hat{p}_{x}:=1$, if ${}_{s}\hat{p}_{x}=0$.
Since $\{ 1,2,\ldots, \lfloor T_{(k)} \rfloor \} \subset [0,T_{(k)}) \subset [0,T_{(k+1)})$, it follows
\[
\begin{split}
\inf_{s \in [0,T_{(k)}) } \frac{{}_{s}p_{x}}{{}_{s}\hat{p}_{x}} \qquad &\leq \qquad \inf_{s \in \{ 1,2,\ldots, \lfloor T_{(k)} \rfloor \} } \frac{{}_{s}p_{x}}{{}_{s}\hat{p}_{x}},
\\
\sup_{s \in [0,T_{(k+1)}) } \frac{{}_{s}p_{x}}{{}_{s}\hat{p}_{x}} \qquad &\geq \qquad \sup_{s \in \{ 1,2,\ldots, \lfloor T_{(k)} \rfloor  \} } \frac{{}_{s}p_{x}}{{}_{s}\hat{p}_{x}},
\end{split}
\]
and so
\[
\begin{split}
\inf_{s \in [0,T_{(k)}) } \frac{{}_{s}p_{x}}{{}_{s}\hat{p}_{x}} \geq 1-\varepsilon_1 \qquad &\Rightarrow \qquad \inf_{s \in \{ 1,2,\ldots, \lfloor T_{(k)} \rfloor  \} } \frac{{}_{s}p_{x}}{{}_{s}\hat{p}_{x}} \geq 1-\varepsilon_1,
\\
\sup_{s \in [0,T_{(k+1)}) } \frac{{}_{s}p_{x}}{{}_{s}\hat{p}_{x}} \leq 1+\varepsilon_2 \qquad &\Rightarrow \qquad \sup_{s \in \{ 1,2,\ldots, \lfloor T_{(k)} \rfloor  \} } \frac{{}_{s}p_{x}}{{}_{s}\hat{p}_{x}} \leq 1+\varepsilon_2.
\end{split}
\]
Therefore
\[
\begin{split}
& \left[ \inf_{s \in [0,T_{(k)}) } \frac{{}_{s}p_{x}}{{}_{s}\hat{p}_{x}} \geq 1-\varepsilon_1 \right]  \cap \left[ \sup_{s \in [0,T_{(k+1)}) } \frac{{}_{s}p_{x}}{{}_{s}\hat{p}_{x}} \leq 1+\varepsilon_2 \right] 
\\
\subset
& \left[ \inf_{s \in \{ 1,2,\ldots, \lfloor T_{(k)} \rfloor \} } \frac{{}_{s}p_{x}}{{}_{s}\hat{p}_{x}} \geq 1-\varepsilon_1 \right] \cap  \left[ \sup_{s \in \{ 1,2,\ldots, \lfloor T_{(k)} \rfloor \} } \frac{{}_{s}p_{x}}{{}_{s}\hat{p}_{x}} \leq 1+\varepsilon_2 \right]
.
\end{split}
\]
In summary,
\begin{equation} \label{EQNfirstthmiii}
\begin{split}
& \left[ \inf_{s \in [0,T_{(k)}) } \frac{{}_{s}p_{x}}{{}_{s}\hat{p}_{x}} \geq 1-\varepsilon_1 \right]  \cap \left[ \sup_{s \in [0,T_{(k+1)}) } \frac{{}_{s}p_{x}}{{}_{s}\hat{p}_{x}} \leq 1+\varepsilon_2 \right]
\\
\subset & \left[ (1+\varepsilon_2) C(0) \geq C(s) \geq (1-\varepsilon_1) C(0)\;\;\mbox{for all $s \in \{1,2,\ldots, \lfloor T_{(k)} \rfloor \}$} \right].
\end{split}
\end{equation}

Now the goal is to write the left-hand side of (\ref{EQNfirstthmiii}) in terms of the order statistics of independent, standard uniformly distributed random variables.   This is done by partitioning $[0,T_{(k)})$ into intervals $[T_{(i-1)}, T_{(i)})$ and considering the minimum value of ${}_{s}p_{x} / {}_{s}\hat{p}_{x}$ over each sub-interval $[T_{(i-1)}, T_{(i)})$.  

The empirical distribution function of the proportion of the initial membership who have died up to time $t \geq 0$ is defined as
\[
\hat{F}_{N} (t) := \frac{1}{N} \sum_{i=1}^{N} \mathbbm{1}_{\{T_i \leq t\}} = \frac{1}{N} \max\{i : T_{(i)} \leq t \}.
\]
Denote by $F$ the distribution function of the death times $T_{1}, T_{2}, \ldots, T_{N}$.  Let $(1-F(s))/(1-\hat{F}_{N}(s)):=1$ if $1-\hat{F}_{N}(s)=0$.  It follows immediately from the definition of the (empirical) survival probability that ${}_{s}\hat{p}_{x} = 1-\hat{F}_{N}(s)$ and ${}_{s}p_{x} = 1-F(s)$. Hence

\begin{equation} \label{EQNsecondthmiii}
\begin{split}
& \left[ \inf_{s \in [0,T_{(k)}) } \frac{{}_{s}p_{x}}{{}_{s}\hat{p}_{x}} \geq 1-\varepsilon_1 \right] \cap \left[ \sup_{s \in [0,T_{(k+1)}) } \frac{{}_{s}p_{x}}{{}_{s}\hat{p}_{x}} \leq 1+\varepsilon_2 \right] 
\\
 = & \left[ \inf_{s \in [0,T_{(k)}) } \frac{1-F(s)}{1-\hat{F}_{N}(s)} \geq 1-\varepsilon_1 \right] \cap \left[ \sup_{s \in [0,T_{(k+1)}) } \frac{1-F(s)}{1-\hat{F}_{N}(s)} \leq 1+\varepsilon_2 \right].
\end{split}
\end{equation}

Let $T_{(0)}:=0$. As the joint distribution of $T_{1}, T_{2}, \ldots, T_{N}$ is continuous, the set $[T_{(i-1)}<T_{(i)}, \, \textrm{for $i= 1,2,\ldots, N+1$}]$ has measure one.  In the following, we work on this set only.

Consider an arbitrary time interval, $[T_{(i-1)}, T_{(i)})$.  The empirical distribution function $\hat{F}_{N}$ changes value only when a member dies, namely only at times $T_{(1)}, \ldots, T_{(N)}$.  In particular, for all $s \in [T_{(i-1)}, T_{(i)})$, $\hat{F}_{N}(s)=\hat{F}_{N}(T_{(i-1)})=\hat{F}_{N}(T_{(i)}-)$, the left limit of $\hat{F}_{N}$ at $T_{(i)}$.

As the distribution function $F$ of the death time is an increasing function, it follows that $1-F(s)$ is a decreasing function that approaches its infimum over $s \in[T_{(i-1)},  T_{(i)})$ at the end of the interval.  Thus, as $\hat{F}_{N}(s)$ is constant over the same interval, the fraction $(1-F(s))/(1-\hat{F}_{N}(s))$ approaches its infimum over $s \in [T_{(i-1)},  T_{(i)})$ as $s$ approaches $T_{(i)}$.  By the continuity of $F$, $F(T_{(i)}-)=F(T_{(i)})$ so that
\[
\inf_{s \in [T_{(i-1)},  T_{(i)}) } \frac{1-F(s)}{1-\hat{F}_{N}(s)} = \frac{1-F(T_{(i)}-)}{1-\hat{F}_{N} (T_{(i)}-)} = \frac{1-F(T_{(i)})}{1-\hat{F}_{N} (T_{(i-1)})}.
\]

On the other hand, as $1-F(s)$ is a decreasing function, it attains its largest value over $s \in [T_{(i-1)}, T_{(i)})$ at the start of the interval. Thus the fraction $(1-F(s))/(1-\hat{F}_{N}(s))$ attains its largest value over $s \in [T_{(i-1)},  T_{(i)})$ at $s=T_{(i-1)}$, i.e.
\[
\sup_{s \in [T_{(i-1)},  T_{(i)}) } \frac{1-F(s)}{1-\hat{F}_{N}(s)} = \frac{1-F(T_{(i-1)})}{1-\hat{F}_{N} (T_{(i-1)})}.
\]
The argument above shows that the infimum of $(1-F(s))/(1-\hat{F}_{N}(s))$ over $s \in [0,T_{(k)})$ is equal to the infimum of $(1-F(T_{(i)}))/(1-\hat{F}_{N} (T_{(i-1)}))$ over $i\in \{1,2,\ldots, k\}$ and that the supremum of $(1-F(s))/(1-\hat{F}_{N}(s))$ over $s \in [0,T_{(k+1)})$ is equal to the supremum of $(1-F(T_{(i-1)}))/(1-\hat{F}_{N} (T_{(i-1)}))$ over $i\in \{1,2,\ldots, k+1\}$.  In summary,
\begin{equation} \label{EQNthirdthm}
\left[ \inf_{s \in [0,T_{(k)}) } \frac{1-F(s)}{1-\hat{F}_{N}(s)} \geq 1-\varepsilon_1 \right] = \left[ \inf_{i \in \{1,2,\ldots, k\}} \frac{1-F(T_{(i)})}{1-\hat{F}_{N} (T_{(i-1)})} \geq 1-\varepsilon_1 \right],
\end{equation}
and, using $T_{(0)}:=0$,
\begin{equation} \label{EQNthirdthmiii}
\begin{split}
\left[ \sup_{s \in [0,T_{(k+1)}) } \frac{1-F(s)}{1-\hat{F}_{N}(s)} \leq 1+\varepsilon_2 \right] = & \left[ \sup_{i \in \{1,2,\ldots, k+1\}} \frac{1-F(T_{(i-1)})}{1-\hat{F}_{N} (T_{(i-1)})} \leq 1+\varepsilon _2\right] \\
= & \left[ \sup_{i \in \{2,3,\ldots, k+1\}} \frac{1-F(T_{(i-1)})}{1-\hat{F}_{N} (T_{(i-1)})} \leq 1+\varepsilon_2 \right] \\
= & \left[ \sup_{i \in \{1,2,\ldots, k\}} \frac{1-F(T_{(i)})}{1-\hat{F}_{N} (T_{(i)})} \leq 1+\varepsilon_2 \right].
\end{split}
\end{equation}

As $F$ is continuous, the random variables $U_{i}=F(T_{i})$ for $i \in \{1,2,\ldots,N\}$ are independent and standard uniformly distributed.  Their order statistics, $(U_{(i)})_{i=1}^{N}$, are linked to the order statistics of the future lifetime random variables by the identity
\[
 U_{(i)}=F(T_{(i)})\quad\mbox{for $i = 1, 2, \ldots, N$}.
\]
Furthermore, by the definition of the empirical distribution function, $\hat{F}_{N} (T_{(i-1)}) = (i-1)/N$, for $i=1,2,\ldots,N$ with $\hat{F}_{N} (T_{(0)}) = \hat{F}_{N}(0) = 0$.  Let $(1-U_{(N)})/0:=1$, then
    \[
    \begin{split}
    & \left[ \inf_{i \in \{1,2,\ldots, k\}} \frac{1-F(T_{(i)})}{1-\hat{F}_{N} (T_{(i-1)})} \geq 1-\varepsilon_1 \right] \cap \left[ \sup_{i \in \{1,2,\ldots, k\}} \frac{1-F(T_{(i)})}{1-\hat{F}_{N} (T_{(i)})} \leq 1+\varepsilon_2 \right]
    \\
    = & \left[ \inf_{i \in \{1,2,\ldots, k\}} \frac{1-U_{(i)}}{1-(i-1)/N} \geq 1-\varepsilon_1 \right] \cap \left[ \sup_{i \in \{1,2,\ldots, k\}} \frac{1-U_{(i)}}{1-i/N} \leq 1+\varepsilon_2 \right] 
    \\
    = & \left[ (1-\varepsilon_1)\tfrac{i-1}{N} + \varepsilon_1 \geq U_{(i)} \;\;\mbox{for all $i \in \{1,2,\ldots, k\}$} \right] 
    \\ & \cap \left\{ (1+\varepsilon_2)\tfrac{i}{N} - \varepsilon_2\leq U_{(i)} \;\;\mbox{for all $i \in \{1,2,\ldots, k\}\setminus\{N\}$} \right] \\
    = & \left[ (1-\varepsilon_1)\tfrac{i-1}{N} + \varepsilon_1 \geq U_{(i)} \geq (1+\varepsilon_2)\tfrac{
\min\{i,N-1\}
    }{N} - \varepsilon_2 \;\;\mbox{for all $i \in \{1,2,\ldots, k\}$} \right].
    \end{split}
    \]

Combining the last equation with equations (\ref{EQNfirstthmiii})-(\ref{EQNthirdthmiii}) and taking the probability, the desired result is obtained.
\end{proof}

In the sequel, the focus is on either a symmetric income threshold -- representing a desire to avoid both upside and downside income volatility -- or a lower income threshold only -- representing a desire to avoid downside income volatility. 

\begin{cor} \label{cor:mainII}
Suppose that for $k\in\{1,2,\ldots,N\}$ and $\varepsilon \in (0,1)$,
    \begin{equation}\label{eq:probupperandlowerbound}
     \mathbb{P} \left[ (1-\varepsilon)\tfrac{i-1}{N} + \varepsilon \geq U_{(i)} \geq (1+\varepsilon)\tfrac{\min\{i,N-1\}}{N} - \varepsilon \;\;\mbox{for all $i \in \{1,2,\ldots, k\}$} \right] \geq \beta.
    \end{equation}
Then
    \begin{equation}\label{eq:probupperandlowerthreshold}
 \mathbb{P} \left[  (1+\varepsilon) C(0) \geq C(s) \geq (1-\varepsilon) C(0)\;\;\mbox{for all $s \in \{1,2,\ldots, \lfloor T_{(k)} \rfloor \}$} \right] \geq \beta.
   \end{equation}
\end{cor}
\begin{proof}
Apply Theorem \ref{theorem:main} with $\varepsilon=\varepsilon_1=\varepsilon_2$.
\end{proof}

\begin{cor} \label{cor:main}
Suppose that, for $k\in\{1,2,\ldots,N\}$ and $\varepsilon \in (0,1)$,
\begin{equation}\label{eq:problowerbound}
 \mathbb{P} \left[(1-\varepsilon)\tfrac{i-1}{N} + \varepsilon \geq U_{(i)}\;\;\mbox{for all $i\in\{1,2,\ldots, k\}$}\right] \geq \beta.
\end{equation}
Then
  \begin{equation}\label{eq:problowerthreshold}
 \mathbb{P} \left[ C(s) \geq(1-\varepsilon) C(0)\;\;\mbox{for all $s \in \{1,2,\ldots, \lfloor T_{(k)} \rfloor \}$} \right] \geq \beta.
   \end{equation}
\end{cor}
\begin{proof}
Apply Theorem \ref{theorem:main} with $\varepsilon=\varepsilon_1$, let $\varepsilon_2\uparrow\infty$ and  observe
\[ (1 + \varepsilon_2) \tfrac{i}{N} - \varepsilon_2 = \tfrac{i}{N} - \varepsilon_2 (1 - \tfrac{i}{N}) \longrightarrow -\infty \quad \mbox{for all $i \in \{ 1,2, \ldots, N-1 \}$}.\]
\end{proof}

The motivation for proving the results in Section \ref{subsection:main-theorem} is that, rather than calculating the maximal integer $k:=k_{C}$ satisfying (\ref{eq:probupperandlowerthreshold}), the maximal integer $k:=k_{U}$ satisfying (\ref{eq:probupperandlowerbound}) can instead be calculated.  The same idea applies for (\ref{eq:problowerbound}) and (\ref{eq:problowerthreshold}).  This has some significant benefits.  

The value of $k_{U}$ is independent of the distribution of $T_1, T_2, \ldots, T_N$ which, if $k_{U}$ is close in value to $k_{C}$, allows general conclusions to be drawn as to the ability of the pooled annuity fund to diversify longevity risk.   Additionally, the calculation of $k_{U}$, which requires only the sampling of random variables, is faster than the calculation of $k_{C}$, which requires the simulation of the income process in the pooled annuity fund.  It is straightforward to show that the calculation of $k_{C}$ requires  the order of $(L+N) \times M$ operations whereas $k_{U}$ takes the order of $N \times M$ operations, where $L$ is the number of times that the income is calculated from time 0 to the maximum time of the last death in the fund, in each simulation.  (For example, if the income was calculated monthly, the participants were age 70 at time 0 and lived until at most age 120 then, at most, $L=600$ time-steps are needed in each simulation to calculate the income paid to the surviving participants.)

Therefore, a key question is how close is the lower bound $k_{U}$ to the true value $k_{C}$?  As is shown next for some chosen mortality models, $k_{U}$ is close enough to $k_{C}$ to use it to make conclusions on the pooled annuity fund's ability to diverse idiosyncratic longevity risk.

\subsection{How close is \pmb{$k_{U}$} to \pmb{$k_{C}$}?}\label{subsection:efficacy}

 To examine how close $k_{U}$ is to $k_{C}$, the maximal integers satisfying each of (\ref{eq:probupperandlowerbound}) to (\ref{eq:problowerthreshold}) individually are calculated; the method for each calculation is described in Sections \ref{subsection:calcku} and \ref{subsection:calckc}, and all simulations were carried out in the statistical software package \emph{R}.

Figure \ref{fig:relative-difference-HMD-IFoA}, which details the extent to which the relative difference between them improves as the number of members $N$ initially in the fund increases, shows that $k_{C}$ is quite close to $k_{U}$.  For example, $k_{C}$ is less than 3\% above the value of $k_{U}$ when $N=2\,000$, and the percentage falls to $2$\% when $N=4\,000$ and falls further to about $1$\% when $N=8\,000$.  These values apply for the two considered life tables and for the two different starting ages, $50$ and $70$ years old.  The relative difference falls as $N$ increases, due to $k_{U}/N$ and $k_{C}/N$ converging to one as $N$ increases.

It is seen from Figure \ref{fig:relative-difference-HMD-IFoA} that the younger group of initially $50$-year-olds has a lower relative error than the older group of initially $70$-year-olds.  The reason is a technical one.  Denote the distribution function of $T_{1}$ for the $50$-year-olds by $F_{50}$ and that for the $70$-year-olds by $F_{70}$, with $F_{50} < F_{70}$.  This inequality implies that for any $y \in (0,1)$ and any given time grid $(t_i)_{i}$, there are more indexes $i$ such that $F_{50}(t_i) < y$ than there are indexes $i$ satisfying $F_{70}(t_i) < y$.  Suppose that for a sequence of uniform order statistics $(U_{(k)})_{k=1}^{N}$, the goal is to check if condition (\ref{eq:probupperandlowerbound}) holds until the $k$th member.  To do this, the condition is checked for all indexes $i$ satisfying $F_{50}(t_i) < U_{(k)}$, and similarly for all indexes $i$ satisfying $F_{70}(t_i) < U_{(k)}$.  Since for the $70$-year-olds there are fewer indices $i$ to check than for $50$-year-olds, the condition is more often fulfilled by the $70$-year-olds, yielding a higher maximal $k_{U}$ for the $70$-year-olds.  This is true even though the values $F_{x}(t_i)$ are different for $x=50$ and $x=70$, since the values are close to each other.  The relative differences are magnified further when $N$ is small, due to the smaller numbers involved.  

A closer look at the results displayed in Figure \ref{fig:relative-difference-HMD-50} shows that the value of the threshold parameter $\varepsilon$ is more important than the value of the certainty $\beta$ in determining the goodness of the approximation.   The larger the value of $\varepsilon$, and therefore the wider are the thresholds, the better is the approximation.  For example, when $\varepsilon=0.1$ and $N=2\,000$, $k_{C}$ is at most 1\% higher than $k_{U}$.  However, the relative difference increases to $3$\% when $\varepsilon=0.05$.  The same observations apply to the other plots in Figure \ref{fig:relative-difference-HMD-IFoA}.

Turning to how these relative differences affect the likely time for which the fund can provide a stable, life-long income, Figure \ref{fig:time-difference-HMD-IFoA} indicates that the difference is at most 10 months.  For funds with at least $2\,000$ members, using $k_{U}$ instead of $k_{C}$ understates the likely time for which a stable, life-long income is provided by between 2 to 4 months.  An explanation of the likely time is given in Section \ref{subsection:confidence}.

The conclusion is that, since it is a close lower bound, $k_{U}$ can be used as an approximation to $k_{C}$.  The calculation methods of $k_{U}$ and $k_{C}$ are described next.


\begin{figure}
\captionsetup[subfigure]{justification=centering} %
    \centering
    \begin{subfigure}[b]{0.45\textwidth} 
        \centering
\includegraphics[trim=0 0.5cm 0 1.5cm,clip,width=\textwidth]{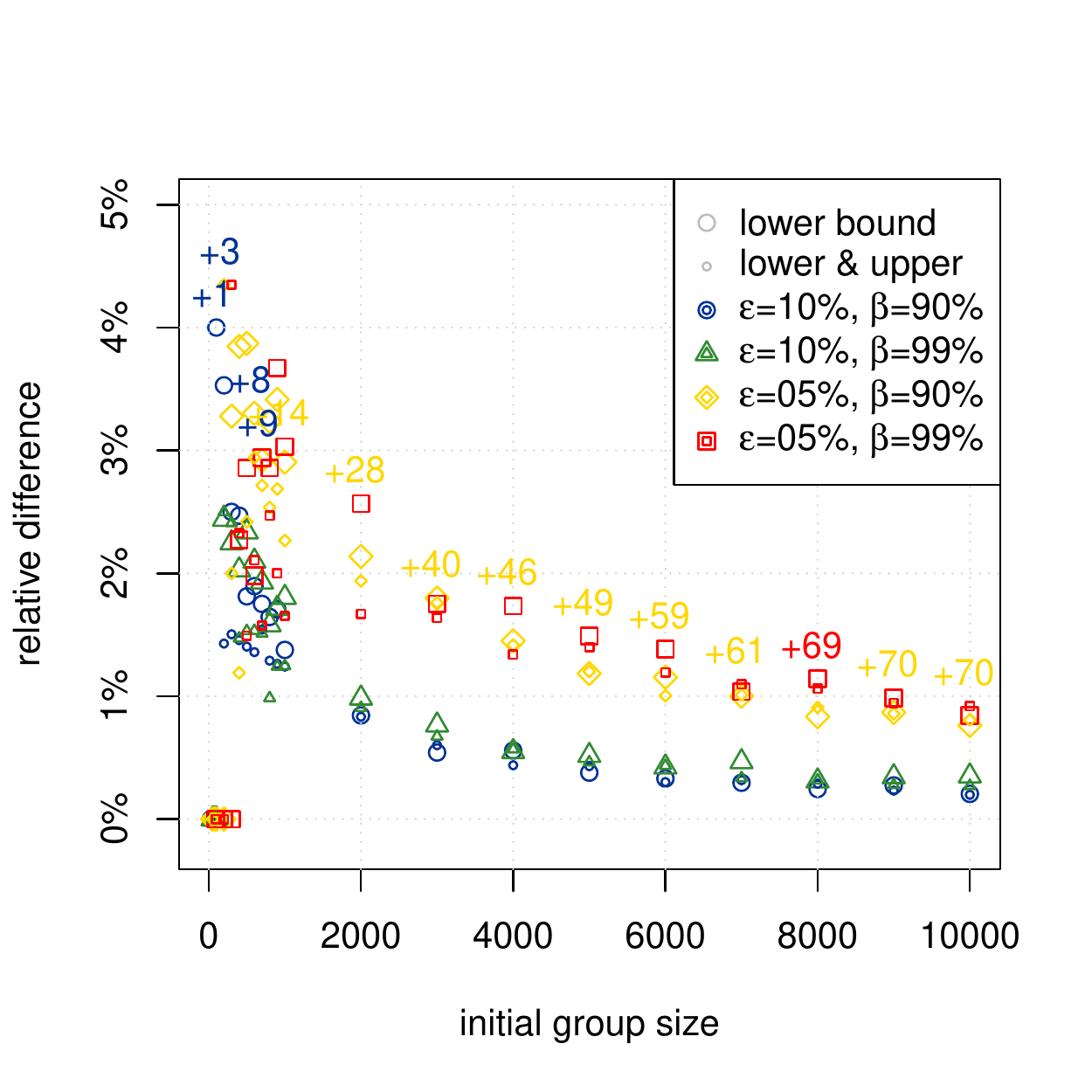}
\caption{Human Mortality Database life table, starting age 50.}
\label{fig:relative-difference-HMD-50}
    \end{subfigure}%
		\begin{subfigure}[b]{0.45\textwidth}
        \centering
\includegraphics[trim=0 0.5cm 0 1.5cm,clip,width=\textwidth]{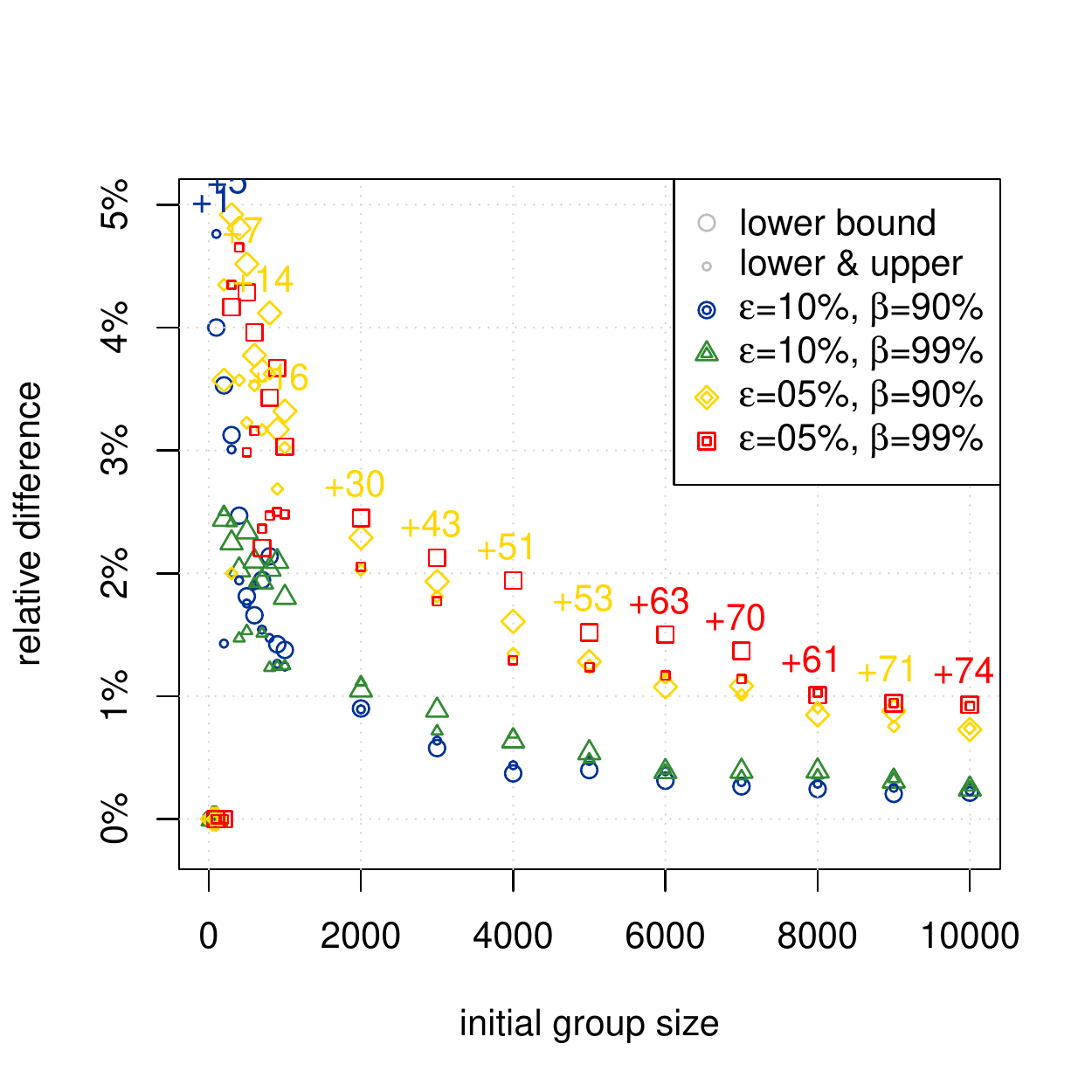}
\caption{Human Mortality Database life table, starting age 70.}
\label{fig:relative-difference-HMD-70}
    \end{subfigure}%
		\vskip\baselineskip 
		    \begin{subfigure}[b]{0.45\textwidth}
        \centering
\includegraphics[trim=0 0.5cm 0 1.5cm,clip,width=\textwidth]{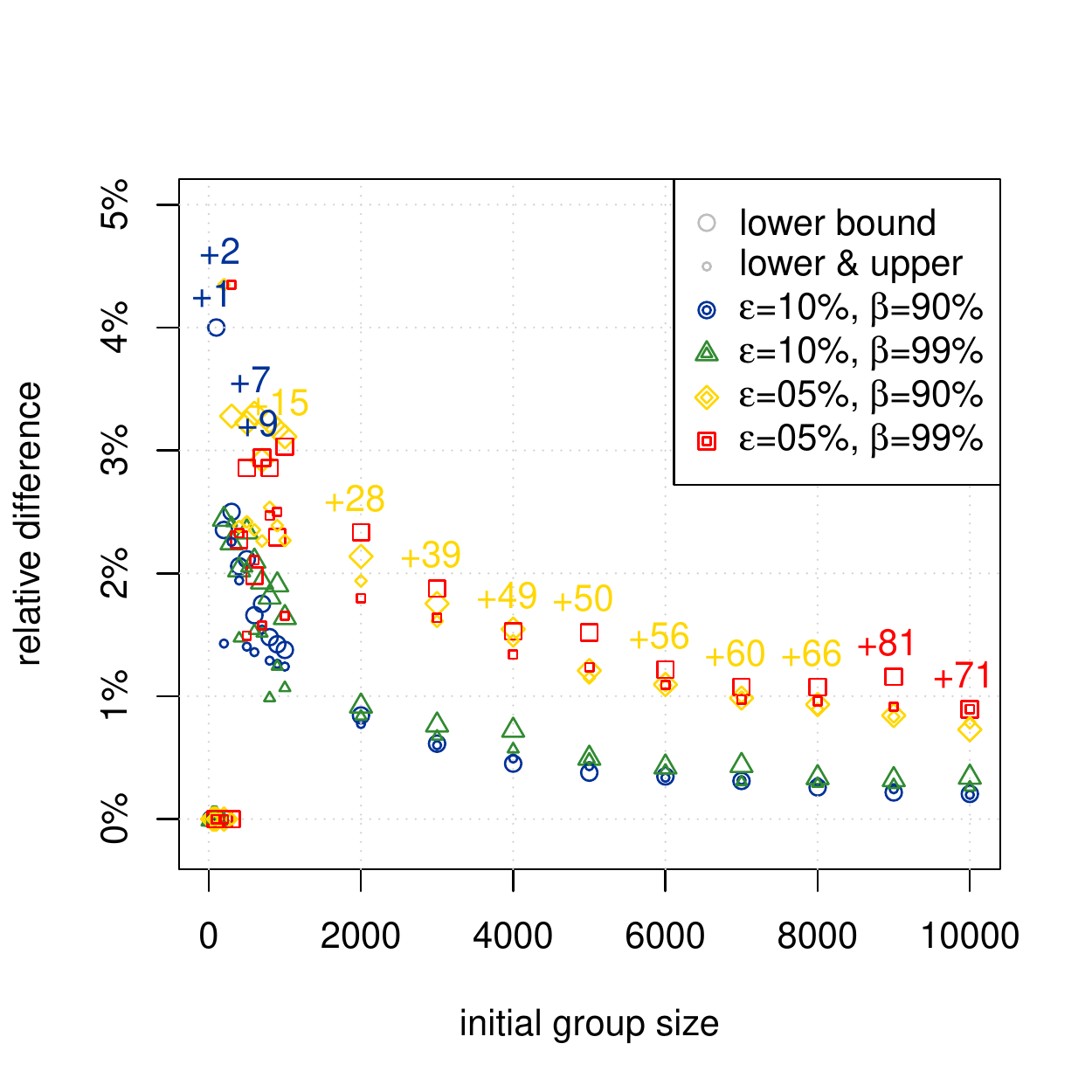}
\caption{Life table S1PFL, starting age 50.}
\label{fig:relative-difference-IFoA-50}
    \end{subfigure}
    \begin{subfigure}[b]{0.45\textwidth}
        \centering
\includegraphics[trim=0 0.5cm 0 1.5cm,clip,width=\textwidth]{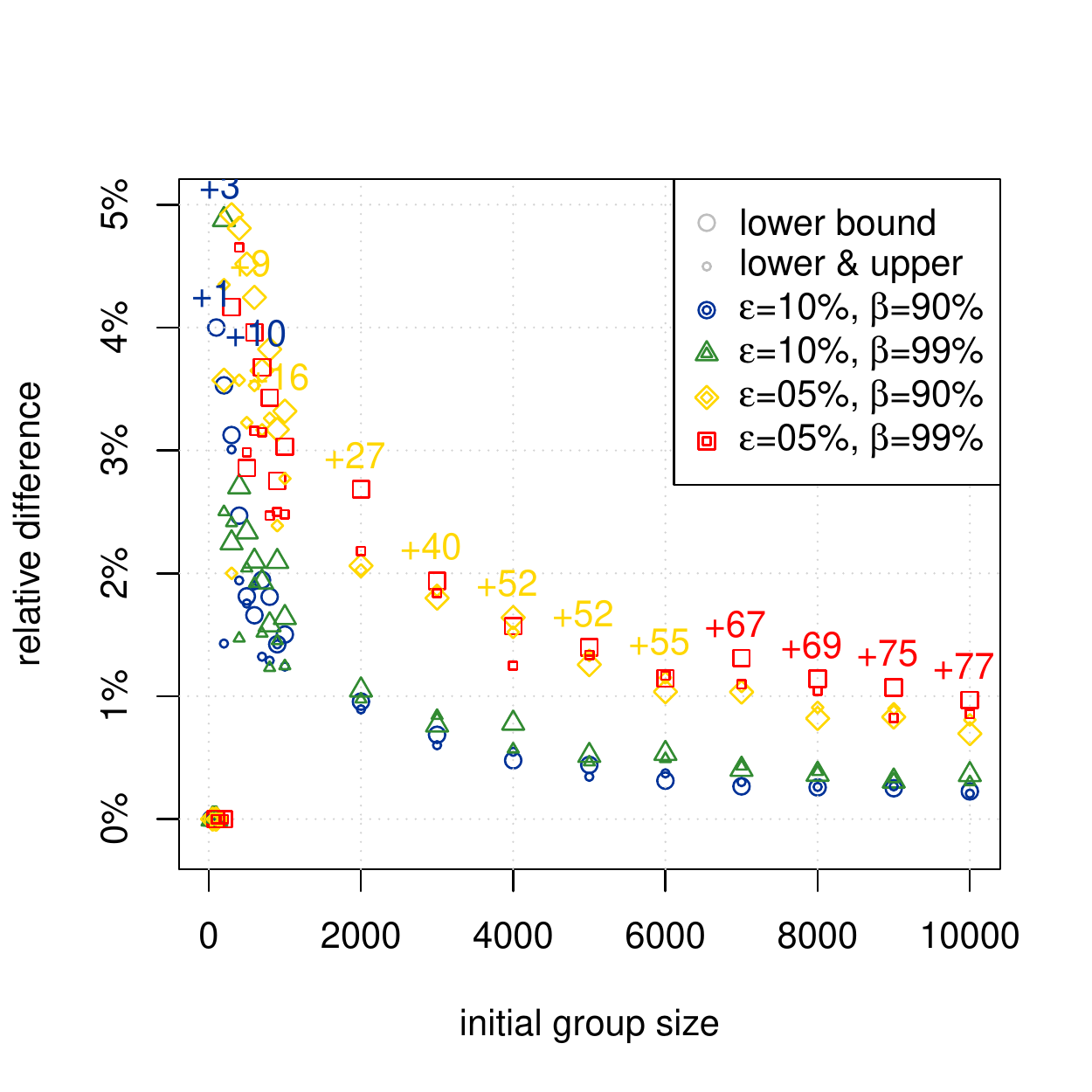}
\caption{Life table S1PFL, starting age 70.}
\label{fig:relative-difference-IFoA-70}
    \end{subfigure}
    \caption{Relative distance of the maximal integer $k:=k_{C}$ fulfilling (\ref{eq:probupperandlowerthreshold}) from the maximal integer $k:=k_{U}$ fulfilling (\ref{eq:probupperandlowerbound}), i.e. $(k_{C}-k_{U})/k_{U}$, as the initial group size increases.  The smallest initial group size is 10.  The relative distances are displayed for a selection of values of $\varepsilon$ and $\beta$, and are indicated by the smaller symbols.  The same calculation is shown for (\ref{eq:problowerthreshold}) and (\ref{eq:problowerbound}), respectively, and indicated by the larger symbols.  The numbers in the plot state the maximum value of $k_{C}-k_{U}$, for each value of $N$, over the eight different combinations of the values of $\beta$ and $\varepsilon$, and whether a lower income threshold only or both income thresholds are included.  For the calculation of $k_{C}$ from (\ref{eq:probupperandlowerthreshold}) and (\ref{eq:problowerthreshold}), it is assumed that income payments are paid monthly to survivors, the initial age is either 50 or 70 years old, and the mortality distribution is based on either the Human Mortality Database' life table for the UK for 2016 \citep{HMD2016} or the UK-based life table S1PFL \citep{IFoA2008}.}
		\label{fig:relative-difference-HMD-IFoA}
\end{figure}


\begin{figure}
\captionsetup[subfigure]{justification=centering} %
    \centering
    \begin{subfigure}[b]{0.45\textwidth} 
        \centering
\includegraphics[trim=0 0.5cm 0 1.5cm,clip,width=\textwidth]{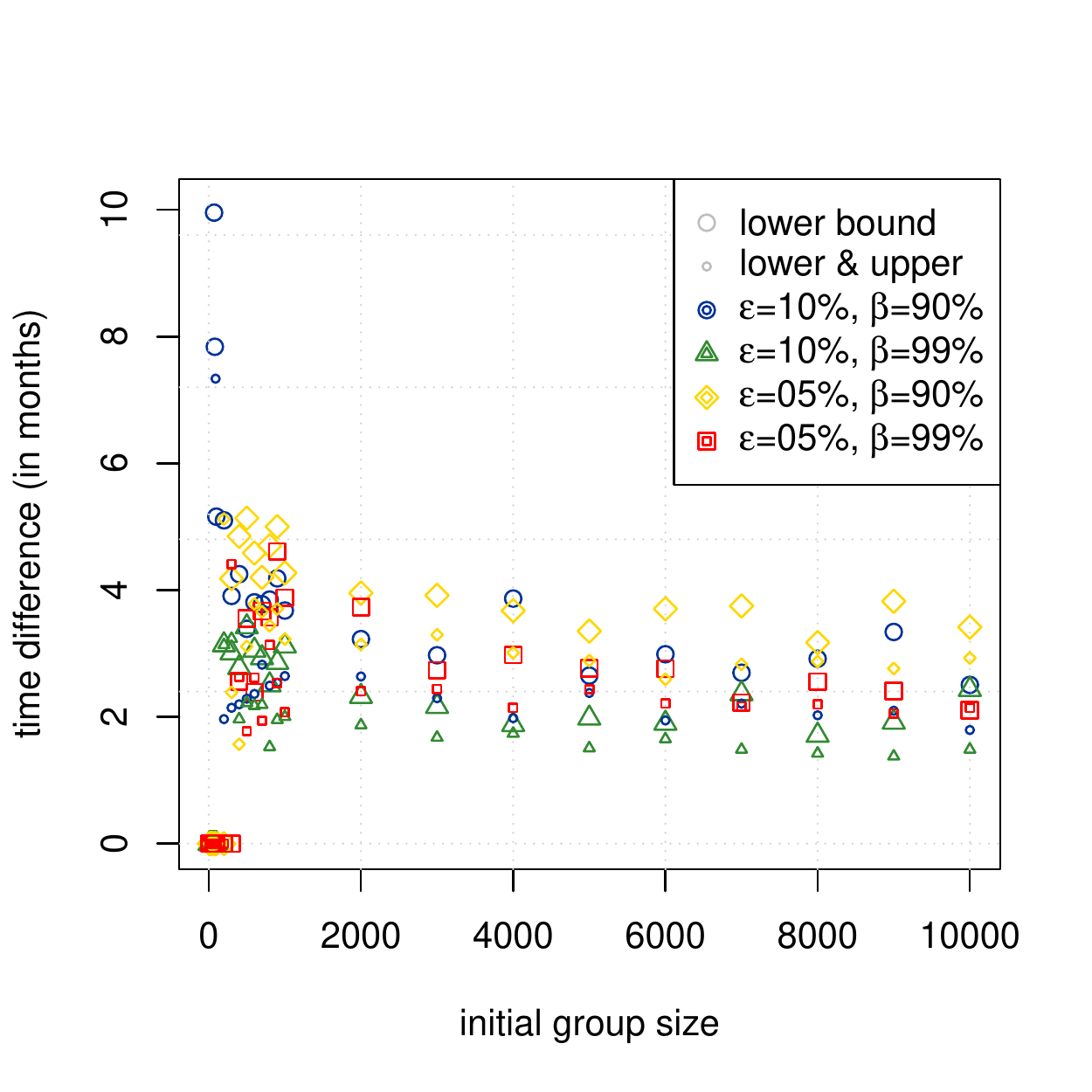}
\caption{Human Mortality Database life table, starting age 50.}
\label{fig:time-difference-HMD-50}
    \end{subfigure}%
		\begin{subfigure}[b]{0.45\textwidth}
        \centering
\includegraphics[trim=0 0.5cm 0 1.5cm,clip,width=\textwidth]{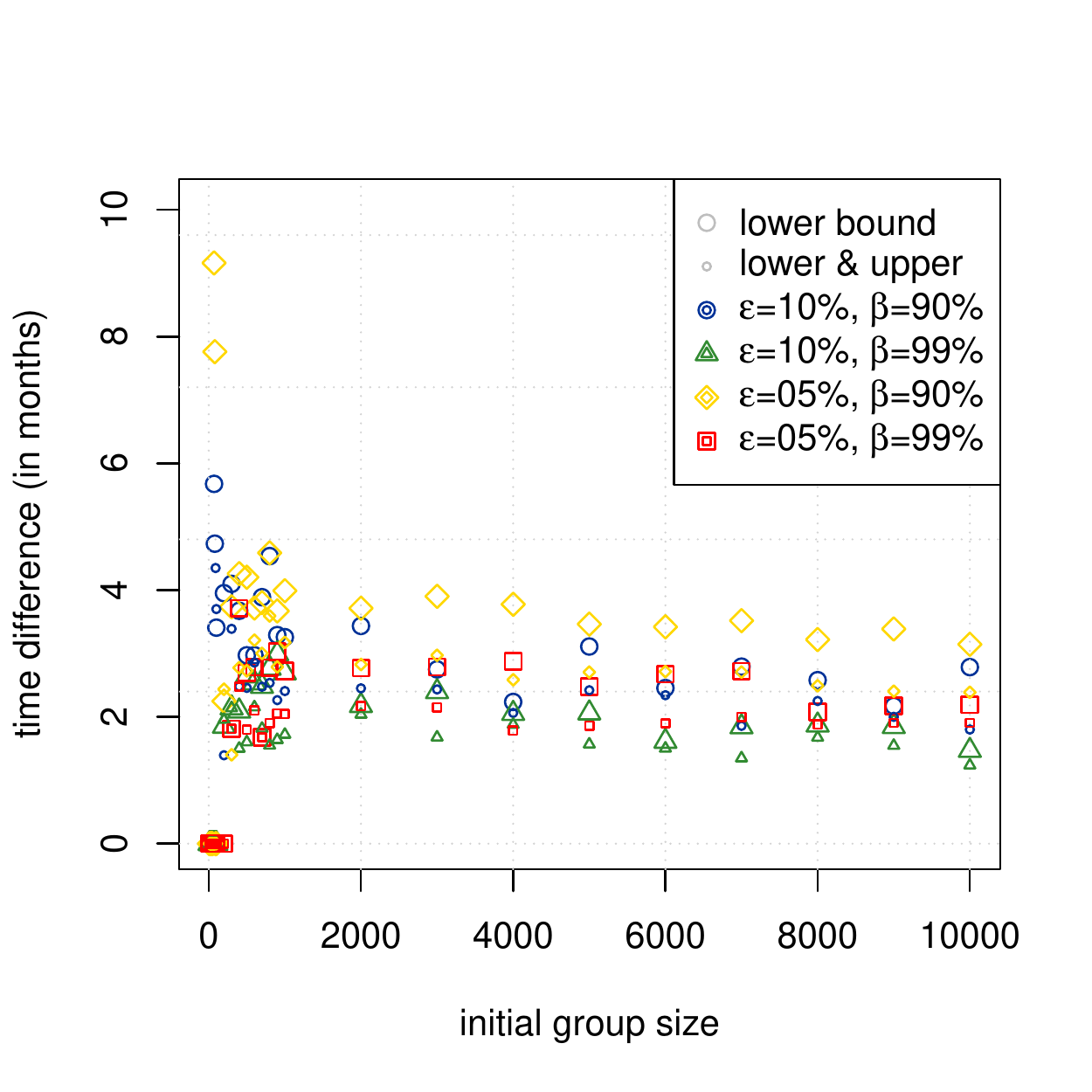}
\caption{Human Mortality Database life table, starting age 70.}
\label{fig:time-difference-HMD-70}
    \end{subfigure}%
		\vskip\baselineskip 
		    \begin{subfigure}[b]{0.45\textwidth}
        \centering
\includegraphics[trim=0 0.5cm 0 1.5cm,clip,width=\textwidth]{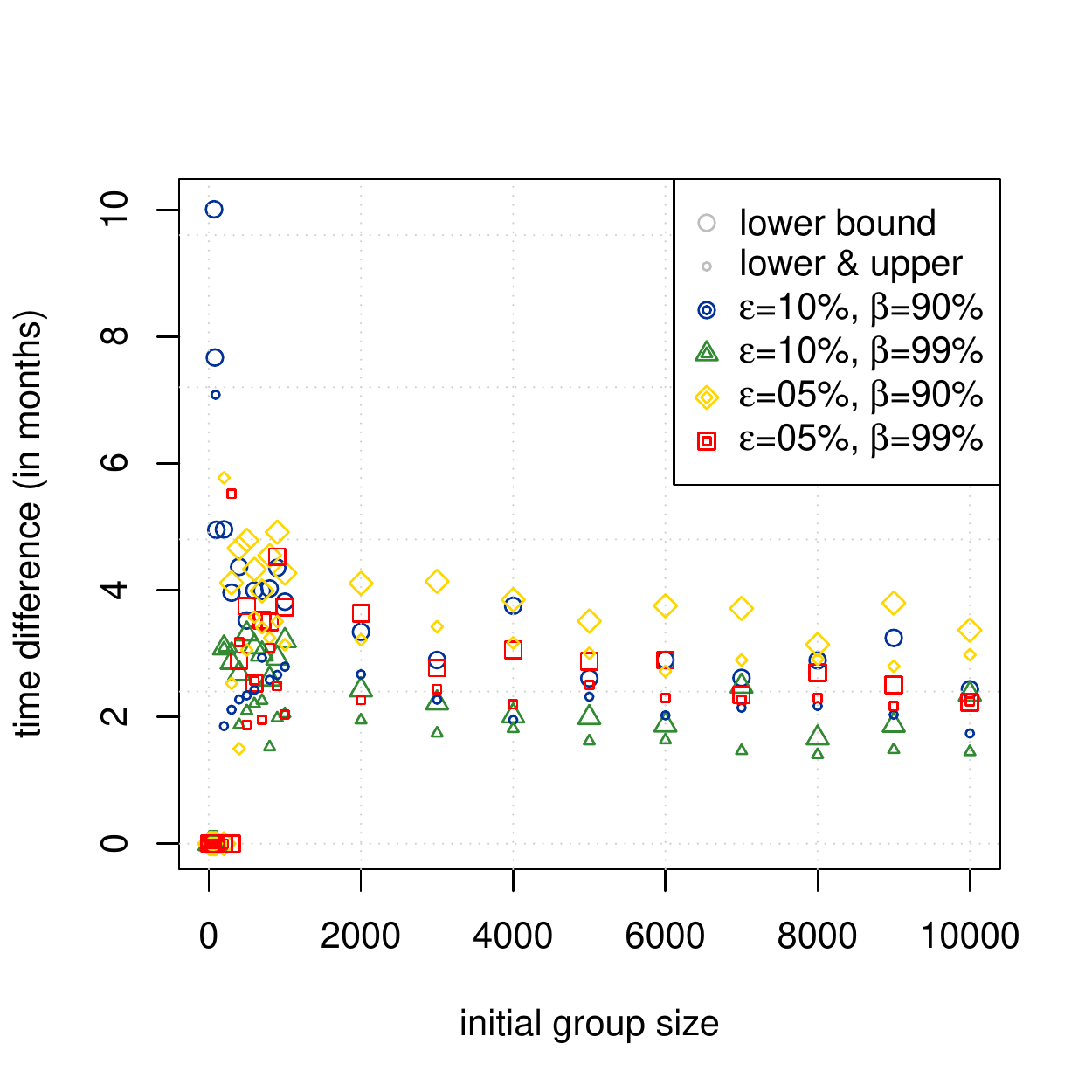}
\caption{Life table S1PFL, starting age 50.}
\label{fig:time-difference-IFoA-50}
    \end{subfigure}
    \begin{subfigure}[b]{0.45\textwidth}
        \centering
\includegraphics[trim=0 0.5cm 0 1.5cm,clip,width=\textwidth]{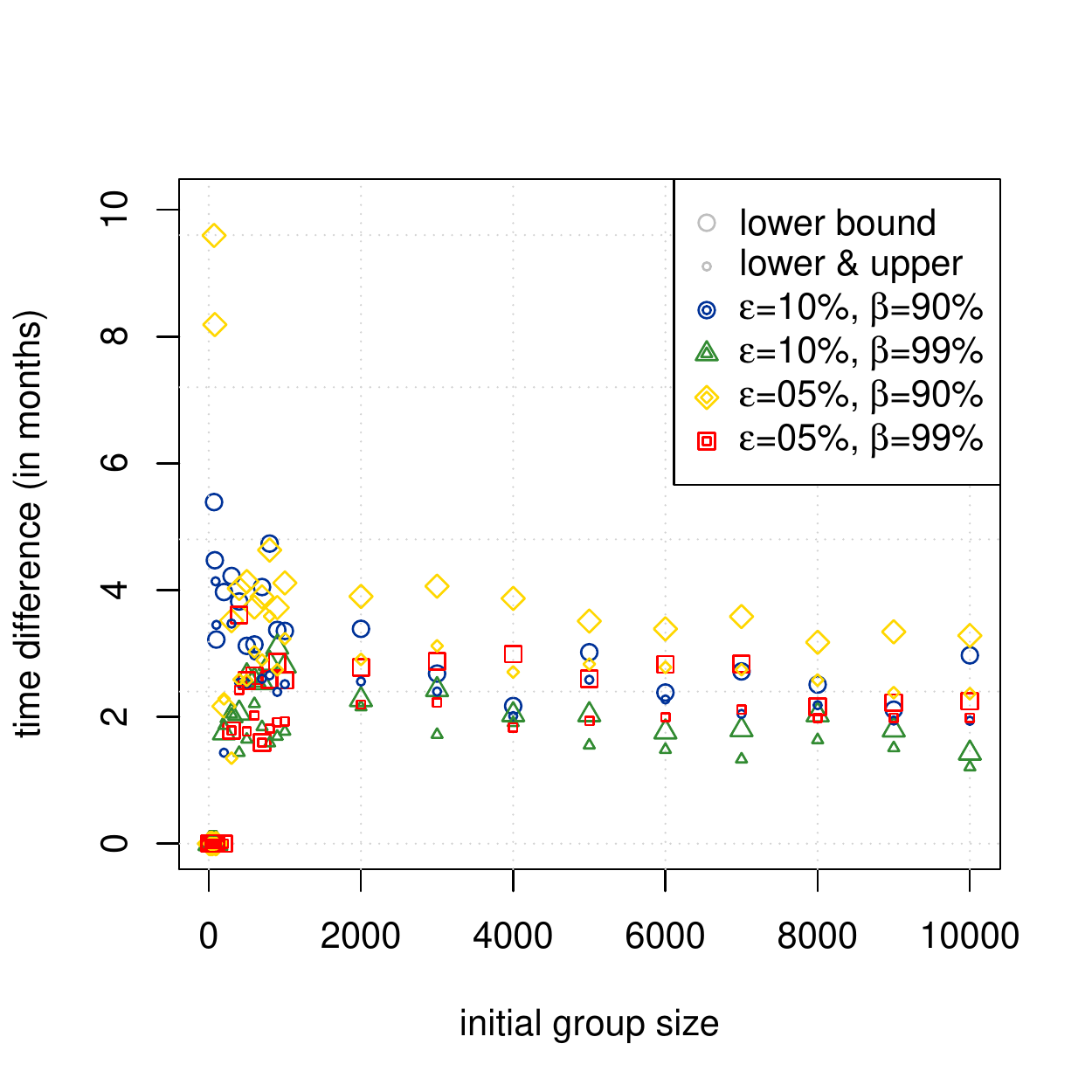}
\caption{Life table S1PFL, starting age 70.}
\label{fig:time-difference-IFoA-70}
    \end{subfigure}
    \caption{Additional likely time for which the fund can provide a stable, life-long income when using the maximal integer $k:=k_{C}$ instead of the maximal integer $k:=k_{U}$, as the initial group size increases.  Thus using $k:=k_{U}$ understates the time for which the fund can provide a stable, life-long income by at most 10 months, for the considered life tables and parameters.  The assumptions are the same as in Figure \ref{fig:relative-difference-HMD-IFoA}.  An explanation of the likely time is given in Section \ref{subsection:confidence}.}
		\label{fig:time-difference-HMD-IFoA}
\end{figure}


\subsubsection{Calculation of the maximal number \pmb{$k_{U}$}}\label{subsection:calcku}
To find the maximal number $k:=k_{U}$ that fulfills either (\ref{eq:probupperandlowerbound}) or (\ref{eq:problowerbound}), Monte Carlo simulation is used.  The initial number of members in the fund, $N$, and the certainty $\beta$ are fixed.  As the income thresholds are symmetric about the initial income, call $\varepsilon:=\varepsilon_{1}=\varepsilon_{2}$ the threshold parameter.  Suppose that $M \in \mathbb{N}$ sample vectors are generated of the uniform order statistics $(U_{(1)}, U_{(2)}, \ldots, U_{(N)})$.

Denote the $m$th sample of the uniform order statistics by $(u^{(m)}_{(1)}, u^{(m)}_{(2)}, \ldots, u^{(m)}_{(N)})$, for $m \in \{1,2,\ldots,M\}$.  For each $m\in \{1,2,\ldots,M\}$, the first integer $i:=i(m) \in \{1, 2, \ldots, N\}$ that fails
\[
	\varepsilon+(1-\varepsilon)(i-1)/N \geq u_{(i)}^{(m)} \geq (1+\varepsilon)\min\{i,N-1\}/N - \varepsilon
\]
is determined and $k(m)=i(m)-1$ is recorded.  If there is no such $i(m)$, then $k(m)=N$ is recorded.

The values $(k(m))_{m=1}^{M}$ are considered as samples from a random variable $K$ and are used to calculate the empirical distribution of $K$.  
Finally, the value of $k_{U}$ is calculated as the $\beta$-quantile of this empirical distribution of $K$.  The same method is used to find the maximal number that fulfills (\ref{eq:problowerbound}).  The final calculated value, $k_{U}$, is the result of $10$ million simulations.

Note that the procedure can be implemented efficiently since the sorting of uniform random variables can be avoided.  According to \citet[f.207]{Devroye1986}, uniform order statistics are ratios of sums of exponential random variables.  More precisely, let $(E_i)_{i=1}^{N+1}$ be independent $\mathrm{Exp}(1)$-distributed random variables and define $S_i=\sum_{j=1}^i E_j$ for integers $\,1\leq j\leq N+1$.  Then the distributions of $(U_{(i)})_{i=1}^N$ and $(S_i/S_{N+1})_{i=1}^N$ are the same.

Table \ref{table:kDelta} lists the maximal number $k_{U}$ as the number of initial members $N$ increases, for a selection of values of $\varepsilon$ and $\beta$, depending on whether either a lower income threshold or a symmetric lower and upper income threshold are applied.  Plotting the maximal numbers (Figure \ref{fig:k-graph}), it can be seen that the value of $k_{U}$ increases approximately linearly for $N \geq 2\,000$.  The linear approximation improves as the income thresholds widen (i.e. as $\varepsilon$ increases) and as the certainty $\beta$ decreases.

\begin{table}
\begin{center}
\begin{tabular}{|r|rr|rr|rr|rr|}
  \hline
  & \multicolumn{2}{c|}{$\varepsilon=10\%$} & \multicolumn{2}{c|}{$\varepsilon=10\%$} & \multicolumn{2}{c|}{$\varepsilon=\hphantom{0}5\%$} & \multicolumn{2}{c|}{$\varepsilon=\hphantom{0}5\%$} \\
  & \multicolumn{2}{c|}{$\beta=90\%$} & \multicolumn{2}{c|}{$\beta=99\%$} & \multicolumn{2}{c|}{$\beta=90\%$} & \multicolumn{2}{c|}{$\beta=99\%$} \\
  \hline
  \centering{$N$} & \multicolumn{1}{c}{$\;k_{U}^{\textrm{above}}$} & \multicolumn{1}{c|}{$\;k_{U}^{\textrm{both}}$} & \multicolumn{1}{c}{$\;k_{U}^{\textrm{above}}$} & \multicolumn{1}{c|}{$\;k_{U}^{\textrm{both}}$} & \multicolumn{1}{c}{$\;k_{U}^{\textrm{above}}$} & \multicolumn{1}{c|}{$\;k_{U}^{\textrm{both}}$} & \multicolumn{1}{c}{$\;k_{U}^{\textrm{above}}$} & \multicolumn{1}{c|}{$\;k_{U}^{\textrm{both}}$} \\
  \hline
  100 & 25 & 21 & 9 & 9 & 6 & 6 & 1 & 1 \\
  200 & 85 & 70 & 41 & 40 & 28 & 23 & 9 & 9 \\
  500 & 331 & 285 & 214 & 196 & 155 & 124 & 70 & 67 \\
  1000 & 799 & 725 & 610 & 562 & 483 & 397 & 264 & 242 \\
  2000 & 1778 & 1680 & 1524 & 1436 & 1310 & 1135 & 857 & 779 \\
  3000 & 2770 & 2662 & 2485 & 2377 & 2224 & 1988 & 1599 & 1466  \\
  4000 & 3766 & 3652 & 3463 & 3342 & 3171 & 2894 & 2421 & 2242 \\
  5000 & 4764 & 4645 & 4450 & 4320 & 4137 & 3829 & 3291 & 3072 \\
  6000 & 5762 & 5641 & 5440 & 5304 & 5113 & 4781 & 4192 & 3940 \\
  7000 & 6761 & 6638 & 6434 & 6292 & 6093 & 5744 & 5112 & 4831 \\
  8000 & 7760 & 7636 & 7427 & 7283 & 7079 & 6715 & 6049 & 5742 \\
  9000 & 8759 & 8634 & 8424 & 8276 & 8067 & 7692 & 6997 & 6670 \\
  10000 & 9758 & 9632 & 9420 & 9269 & 9059 & 8673 & 7952 & 7608 \\ 
  \hline
\end{tabular}
\caption{The values $k_{U}^{\textrm{above}}$ and $k_{U}^{\textrm{both}}$ are the maximal integer $k$ that satisfies (\ref{eq:problowerbound}) and (\ref{eq:probupperandlowerbound}), respectively.}  
\label{table:kDelta}
\end{center}  
\end{table}

\begin{figure}
\begin{center}
\includegraphics[clip,trim={0 0.5cm 0 1.25cm},width=0.6\linewidth]{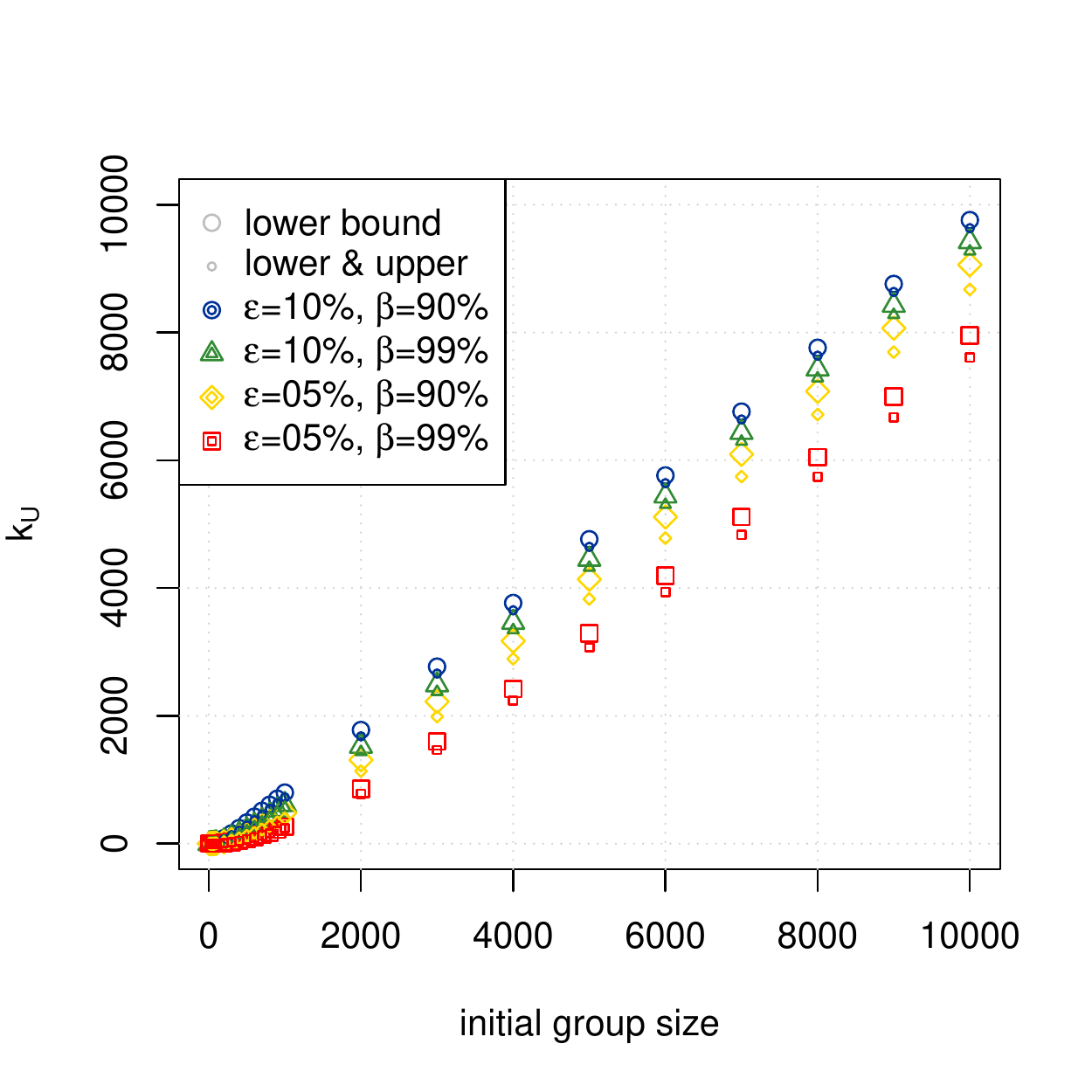}
\captionof{figure}{
The maximal number of members $k_{U}$ for whom the lifetime income is stable, plotted against the initial number of members $N$ in the fund, for a selection of threshold parameters $\varepsilon$ and certainties $\beta$.  The number is calculated from either (\ref{eq:problowerbound}), represented by the larger symbols in the plot, or (\ref{eq:probupperandlowerbound}), represented by the smaller symbols.
}
\label{fig:k-graph}
\end{center}
\end{figure}

\subsubsection{Calculation of the maximal number \pmb{$k_{C}$}}\label{subsection:calckc}

Calculating $k_{C}$ involves a straightforward simulation of $M=10$ million sample paths of the income process, based on monthly time units and applying an appropriate life table.  Along the $m$th sample path, the maximum value of $k(m)$ satisfying the event in inequality (\ref{eq:probupperandlowerthreshold}) or (\ref{eq:problowerthreshold}), as appropriate, is determined.  Similar to the calculation above, the value of $k_{C}$ is calculated as the $\beta$-quantile of the empirical distribution of a random variable $K$, for which $(k(m))_{m=1}^{M}$ are the observed samples.

\subsection{Discussion of the main results}\label{subsection:discuss-theorem}

The purpose of the pooled annuity fund is to pay a stable life-long income to its participants.  What universal features can be noted about its ability to do so, if volatility is due only to idiosyncratic longevity risk and $k_{U}$ is used as a close approximation to $k_{C}$?  In this section, stability is defined as the future income payments each being within $100 \varepsilon$\% of the initial income $C(0)$ (as in Corollary \ref{cor:mainII}).  This is equivalent to considering the values of $i$ for which the condition
\begin{equation} \label{eq:tontine-conditionII}
	\varepsilon+(1-\varepsilon)(i-1)/N \geq U_{(i)} \geq (1+\varepsilon)\min\{i,N-1\}/N - \varepsilon
\end{equation}
holds, in which the upper bound $\varepsilon+(1-\varepsilon)(i-1)/N$ corresponds to the lower income threshold $100(1-\varepsilon)C(0)$.  The lower bound $(1+\varepsilon)\min\{i,N-1\}/N - \varepsilon$ corresponds to the upper income threshold $100(1+\varepsilon)C(0)$.

First, note that the distance between the lower and upper bounds narrows as $i$ increases.  Geometrically, considering the bounds as lines, the gradient to the lower bound is higher than the gradient of the upper bound.  Thus as $i$ increases, it may be more likely that a sample of $U_{(i)}\sim \textrm{Beta}(i,N-i+1)$ \citep[pp.97]{ShorackWellner2009} lies outside of the bounds.  To investigate this possibility further, the bounds and the $0.5$\% and $99.5$\% quantiles of the uniform order statistics are plotted in Figure \ref{fig:eps5q005995UL} for $\varepsilon=0.05$.  Note that although the support of the distribution functions of the uniform order statistics is $[0,1]$, the quantiles indicate the location of the vast majority (99\%) of their possible values.

Figure \ref{fig:eps5q005995UL} illustrates that the income received at the end of the life of the longest-lived members is less likely to lie between the upper and lower income thresholds than for the shortest-lived members.  This is seen from Figure \ref{fig:eps5q005995UL}: the $0.5$\% and $99.5$\% quantile lines (the uppermost and lowermost black, dashed lines in the figure) lie inside the bounds (red, solid lines) for the lower order statistics (i.e. the smaller values of $i/N$, which denote the shorter-lived members).  For larger values of $i/N$ (denoting the longer-lived members), the quantile lines lie outside the bounds in the plot.   When this occurs depends on the number of people initially in the fund.  For $N=200$, the quantiles are seen to fall quickly outside the bounds, as $i/N$ increases (Figure \ref{fig:N200eps5q005995UL}).  Only about the first $6$ members ($i/200 \approx 0.03$) are extremely likely to get a life-long, stable income.  As $N$ increases to $2\,000$ (Figure \ref{fig:N2000eps5q005995UL}), it is only for approximately the first $900$ members (i.e. the first $800$ members; $i/2000 \approx 0.45$) that the $0.5$\% to $99.5$\% quantiles of $U_{(i)}$ falls outside the bounds.   Note that these values are not the same as the maximal number of members who get a stable, life-long income with a certainty level of $99$\%, since they look only at the quantiles at each time of death, rather than along each possible future scenario.

\begin{figure}[H]
\captionsetup[subfigure]{justification=centering} %
    \centering
    \begin{subfigure}[b]{0.5\textwidth}
        \centering
\includegraphics[trim=0 0.5cm 0 1.5cm,clip,width=\textwidth]{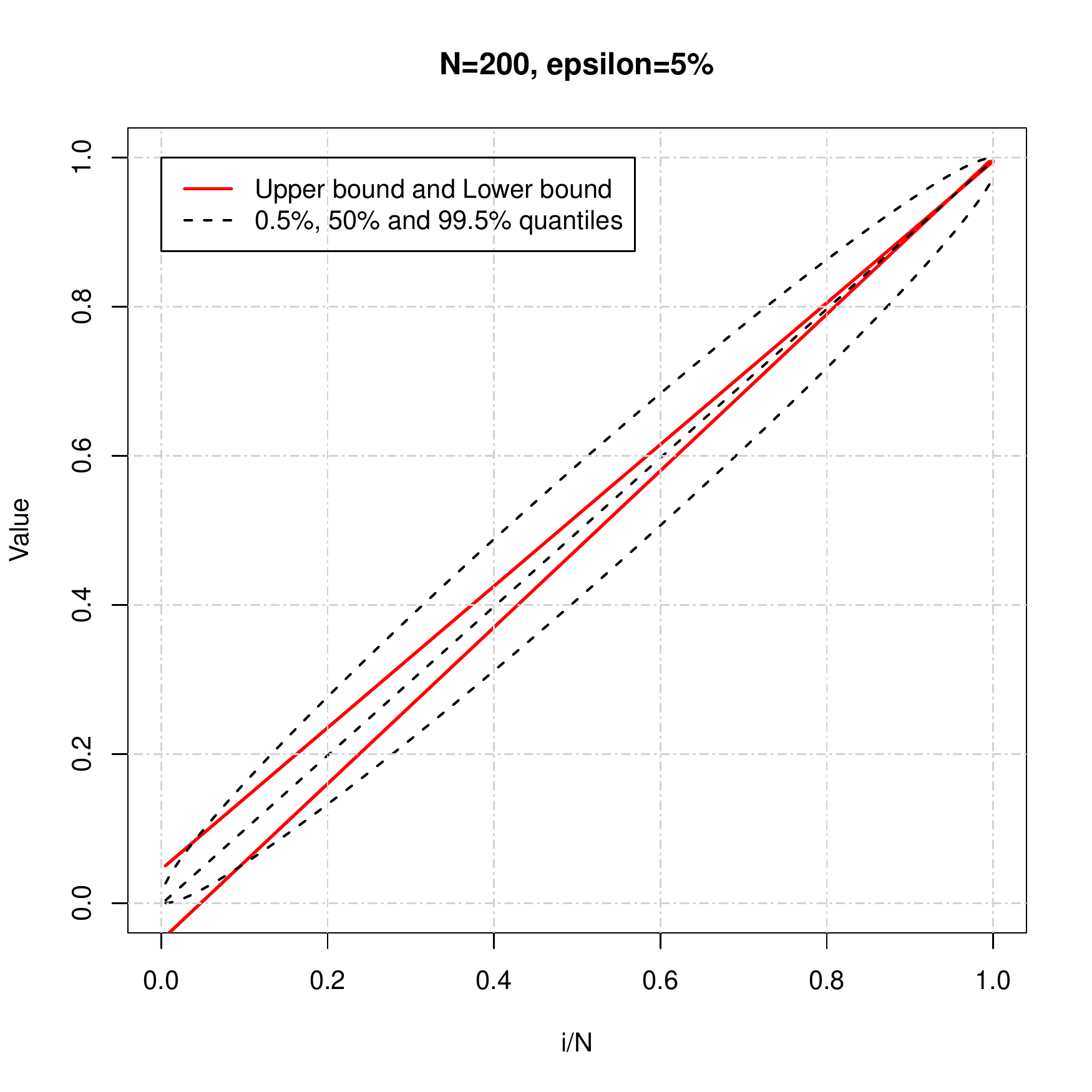}
\caption{$\varepsilon=0.05$ and $N=200$.}
\label{fig:N200eps5q005995UL}
    \end{subfigure}%
    \begin{subfigure}[b]{0.5\textwidth}
        \centering
\includegraphics[trim=0 0.5cm 0 1.5cm,clip,width=\textwidth]{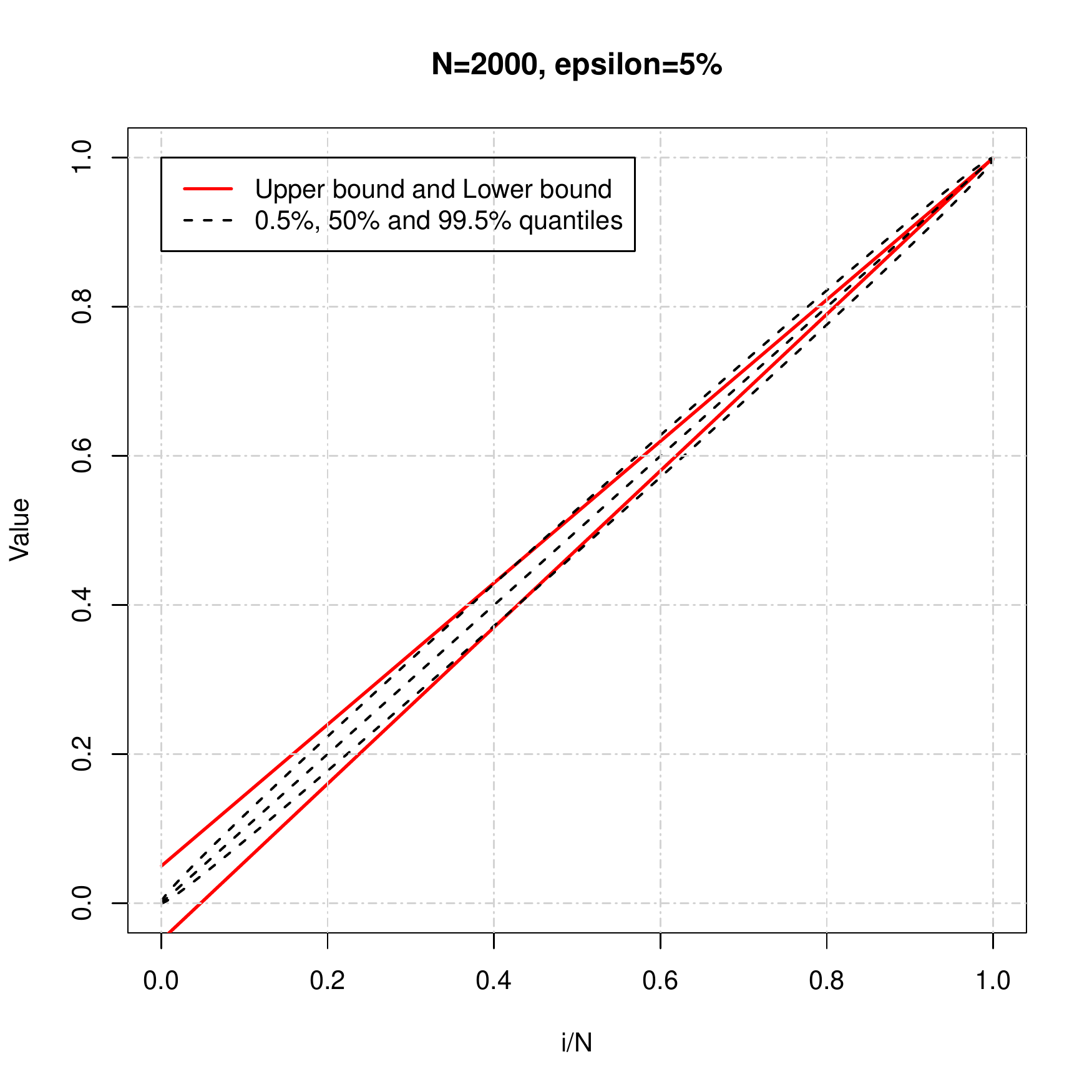}
\caption{$\varepsilon=0.05$ and $N=2\,000$.}
\label{fig:N2000eps5q005995UL}
    \end{subfigure}
    \caption{The symmetric upper and lower bounds on the order statistics, for $\varepsilon=0.05$ and either $N=200$ or $N=2\,000$, shown as solid red lines.  The $0.5$\%, $50$\% and $99.5$\% quantiles of the $i$th uniform order statistic are displayed as dashed black lines.  The $x$-axis values are $i/N$, rather than $i$, so that the domains of the plotted functions are all $[0,1]$.}
		\label{fig:eps5q005995UL}
\end{figure}

The second item to note on the condition (\ref{eq:tontine-conditionII}) is that it fails to hold for the highest values of $i$, by virtue of the lower and upper bounds crossing each other.  The point at which they cross is found by solving 
\[
\varepsilon+(1-\varepsilon)(i-1)/N = (1+\varepsilon) \min\{i,N-1\}/N - \varepsilon
 \]
for $i$.  Setting $i=N$ gives a contradiction ($N = N-1$) so assume that $i<N$ to obtain that the bounds have crossed for each $i$ satisfying
\begin{equation} \label{eq:symmetricboundresult}
i \geq N + \left(1-\varepsilon^{-1} \right)/2.
\end{equation}
Since the bounds are independent of the order statistics, in all future states of the world the following can be observed. 
\begin{itemize}
\item In view of $\varepsilon\leq1$, the right-hand side of the inequality (\ref{eq:symmetricboundresult}) is less than or equal to $N$.  Thus there is always at least one value of $i$ (i.e. $i=N$) for which the bounds have crossed.  Indeed, as $\varepsilon \rightarrow 1$, the right-hand side of the inequality (\ref{eq:symmetricboundresult}) tends to $N$, which means that the two bounds cross only as the last member of the fund dies.

\item As $\varepsilon \rightarrow 0$, the right-hand side of the inequality (\ref{eq:symmetricboundresult}) tends to $-\infty$ and the two bounds cross before any fund member dies.  In the limit, no fund member can receive a stable (i.e. constant) income in any future state of the world, meaning that the choice of $\beta$ is irrelevant.

\item For fixed $\varepsilon$, there are a fixed number of members who are still alive when the bounds have crossed each other.  For example, for $\varepsilon=0.05$ the longest-lived $10$ members do not receive a stable income for the entirety of their future lifetime in any future state of the world, regardless of the initial number of members in the fund (since the bounds have crossed for $i \geq N + (1-\varepsilon^{-1} )/2 = N-9.5$, i.e. $i \in \{N-9, N-8, \ldots, N\}$, a set of $10$ elements.).  The remaining $N-10$ members may or may not receive a life-long income in any particular state of the world.  

\end{itemize}

In summary, the longest-lived members are unlikely to receive a stable, life-long income when a symmetric income threshold is used to define income stability.  This is by dint of the narrowing between the symmetric bounds for the longest-lived members and the distribution of the uniform order statistics relating to those longest-lived members not concentrating at the same rate.   In fact, as long as $\varepsilon<1$, the longest-lived members have no opportunity to get a stable income for the entirety of their future lifetime, in any future state of the world.  Nonetheless, they may still receive a stable income for the majority of their future lifetime.

These observations support the increasing volatility in income for the longest-lived members (in the last cohort to enter if the fund is an open one), remarked upon in \citet{Piggottetal2005}, \citet{qiaosherris2013} and \citet{sabin2010}.  The results in this paper enable the precise elucidation in our model of when ``increasingly volatile'' becomes ``too volatile''.

Next turn to the case when income stability is defined as the income lying above a lower income threshold only.  This is equivalent to considering the values of $i$ for which the condition
\[
	\varepsilon+(1-\varepsilon)(i-1)/N \geq U_{(i)}
\]
holds.  From Figure \ref{fig:eps5q005995UL}, it appears that the 50\% quantile line lies below the upper bound (the upper solid red line in Figure \ref{fig:eps5q005995UL}; the lower solid red line can be ignored as it represents the redundant upper income threshold).  Considering the mean of the $i$th order statistic, $i/(N+1)$, instead of the median -- since it has a simple formula that is amenable to algebra and is not dissimilar in value \citep{Kerman2011} -- then
\[
\varepsilon+(1-\varepsilon)(i-1)/N \geq i/(N+1) \qquad \Rightarrow \qquad \begin{cases}
i \leq N+1 \quad \textrm{if $\varepsilon > 1/(N+1)$} \\
i \geq N+1 \quad \textrm{if $\varepsilon \leq 1/(N+1)$}.
\end{cases}
\]
Thus the means of the $N$ order statistics all lie above the upper bound $\varepsilon+(1-\varepsilon)(i-1)/N$, if the income threshold parameter $\varepsilon \leq 1/(N+1)$.  Approximating the medians of the $N$ order statistics by their means, this means that there is a less than 50\% chance of each order statistic lying below the upper bound.

In summary, the value of $\varepsilon$ should not be too small, i.e. it should be at least above $1/(N+1)$.  Otherwise, life-long income stability will not be achieved for a substantial proportion of the fund membership.  This comment applies in the presence of either a lower income threshold only or both income thresholds.

\section{Application: determining for how long the fund can provide a stable income} \label{subsection:confidence}

The time at which the income of a pooled annuity fund becomes unstable -- i.e. when the income no longer lies between the thresholds in at least $100\beta$\% of future states of the world -- is calculated in terms of the number of deaths that occur up to that time.  However, it is useful to communicate this as a length of time which represents how long all members receive a stable income with a specified level of certainty.  The members who die before that length of time receive a life-long stable income, whereas those who die after it receive a stable income for only part of their lifetime.

Since the time at which a specified number of deaths occurs varies in each future state of the world, an estimate of the time at which it occurs is calculated, which is called the likely time.  This is an approximation of the true time, since the number of deaths used, $k_{U}$, is calculated by maximising the left-hand side of inequality (\ref{eq:problowerbound}) over $k$, rather than maximising that of inequality (\ref{eq:problowerthreshold}).

The calculation of the length of time requires the selection of a mortality law, and the mortality law matters.  For example, the length of time over which the fund is likely to provide a stable income will increase as mortality lightens (i.e. people live longer).  

For example, assume that $N=2\,000$ members all join the pooled annuity fund at age 70.  The annual income $C(m)$ which they withdraw at time $m \in \{0,1,2,\ldots\}$ is calculated by dividing their fund value at time $m$ by the value of a single life annuity calculated at age $70+m$, as described in Section \ref{SUBSECincomeprocess}.  Suppose that the future income is considered stable if it is at least $95\%$ of the initial income (i.e. lower threshold parameter of $\varepsilon=0.05$ and no upper threshold parameter), in $90$\% of all future scenarios (i.e. certainty level $\beta=0.9$).  To determine for how long the fund is likely to provide a stable income, a two-step calculation is required.  
\begin{itemize}
\item First, by Monte Carlo simulation (as described in Section \ref{subsection:calcku}) it is calculated that the highest value of $k$ which satisfies inequality (\ref{eq:problowerbound}) is $k_{U}=1\,310$.  

Then, by inequality (\ref{eq:problowerthreshold}), all members of the fund get an income of at least $0.95C(0)$ up to time $\lfloor T_{(1310)} \rfloor$, in at least $90$\% of future scenarios.  As shown in Section \ref{subsection:efficacy}, the maximal value $k:=k_{C}$ that satisfies inequality (\ref{eq:problowerthreshold}) is about $2$\% higher than the maximal value $k:=k_{U}=1\,310$ that satisfies inequality (\ref{eq:problowerbound}) for $N=2\,000$, for the considered mortality table.  So in reality, $k_{C}\approx 1\,338$ and thus the stable income is provided for longer than suggested by $k_{U}$.  Mitigating the small error caused by this under-estimation is that, in the second step described next, the higher time $T_{(1310)}$ at which $1\,310$ deaths have occurred is used, rather than the time $\lfloor T_{(1310)} \rfloor \leq T_{(1310)}$.

\item Second, calculate at what time is it likely that exactly $k_{U}=1\,310$ deaths have been observed in the initial population of $2\,000$ 70-year-olds.  This is the problem of finding $t$ that satisfies ${}_{t}q_{70} = 1\,310/2\,000$, in which the constant ${}_{t}q_{70}$ is the true probability that a $70$-year old dies between age $70$ and age $70+t$.  Here, $t$ is called the likely time.  It is a straightforward task if one has a life table at hand.  For example, assuming that the mortality of each of the $2\,000$ members initially in the fund follows the UK-based life table S1PFL \citep{IFoA2008} and assuming a uniform distribution of deaths between integer ages, gives that $t=19.4$ years.  
\end{itemize}

Note that only small variations of the time of death of the $k_{U}$th member are expected.   This is because as $N$ becomes larger, the probability density of the $k$th uniform order statistic tends to a Dirac delta function.  To examine how quickly this happens, the probability that the time of the $k_{U}$th death occurs at least $d\geq 0$ time units before its likely time of occurrence $t$, i.e. $\mathbb{P}[T_{(k_{U})} \leq t-d]$ is plotted against $d$ in Figure \ref{fig:beta-Time-lower-upper}, for different choices of $N$, $\varepsilon$ and $\beta$ and for one mortality law.  Symmetric lower and upper income thresholds are imposed.  The figure when only a lower income threshold is imposed is very similar and is omitted.

For an initial number of members $N=100$ in the fund (the thickest lines in Figure \ref{fig:beta-Time-lower-upper}), the time of the $k_{U}$th death has a $99$\% probability of appearing between $2.5$ to $3$ years below $t$.  For groups with initially $N=1\,000$ members, there is a $99\%$ probability that the $k_{U}$th death appears at most one year below $t$.  This falls to $0.5$ years with $N=10\,000$ members.  The conclusion is that the likely time $t$ calculated in the second step above is very close to the actual time of the $k_{U}$th death when there are at least $1\,000$ members initially in the fund, being at most one year greater than the actual time.

Figure \ref{fig:Time-IFoA} shows the likely time $t$ that the income is stable, for the same mortality law.   The time increases with $N$, since more members mean that idiosyncratic longevity risk is more diversified and thus the income paid from each surviving member's account is stable for longer.  Although the rate of increase slows as $N$ increases, the income is stable for quite a long time once the number of initial members $N$ is at least $2\,000$.  For example, when symmetric lower and upper income thresholds are imposed, the likely time for which a stable income can be provided with certainty $\beta$ increases from $17$ years to $25$ years as $N$ increases from $2\,000$ to $10\,000$, for $\varepsilon=0.05$ and $\beta=0.9$.

Now turn to the case when there is only a lower threshold imposed in the definition of income stability.  Starting with $N=2\,000$ members, the income is at least $95$\% of the initial income (i.e. $\varepsilon=0.05$) for nearly $15$ years in $\beta=99$\% of future states of the world.  This lower income stability extends to almost $20$ years if the certainty is dropped to $\beta=90$\%.  From Figure \ref{fig:Time-IFoA} it is seen that including both an upper and a lower threshold reduces by, at most, $2$ years the length of time for which a fund is likely to provide an income above a lower threshold.

It is also seen from the figure that the threshold parameter $\varepsilon$ has a larger effect on the likely times than the certainty $\beta$, for each value of $N$.  For example, for a fund with initially $N=2\,000$ members, decreasing the threshold parameter $\varepsilon$ from $10$\% to $5$\% reduces the likely time by up to $7$ years, everything else being the same.  However, decreasing the certainty level $\beta$ from $99$\% to $90$\% reduces the likely time by up to $5$ years, everything else being the same.  These reductions decline in magnitude as $N$ increases.

\begin{figure}[H]
    \centering
\includegraphics[trim=0 0.5cm 0 1.5cm,clip,width=0.5\textwidth]{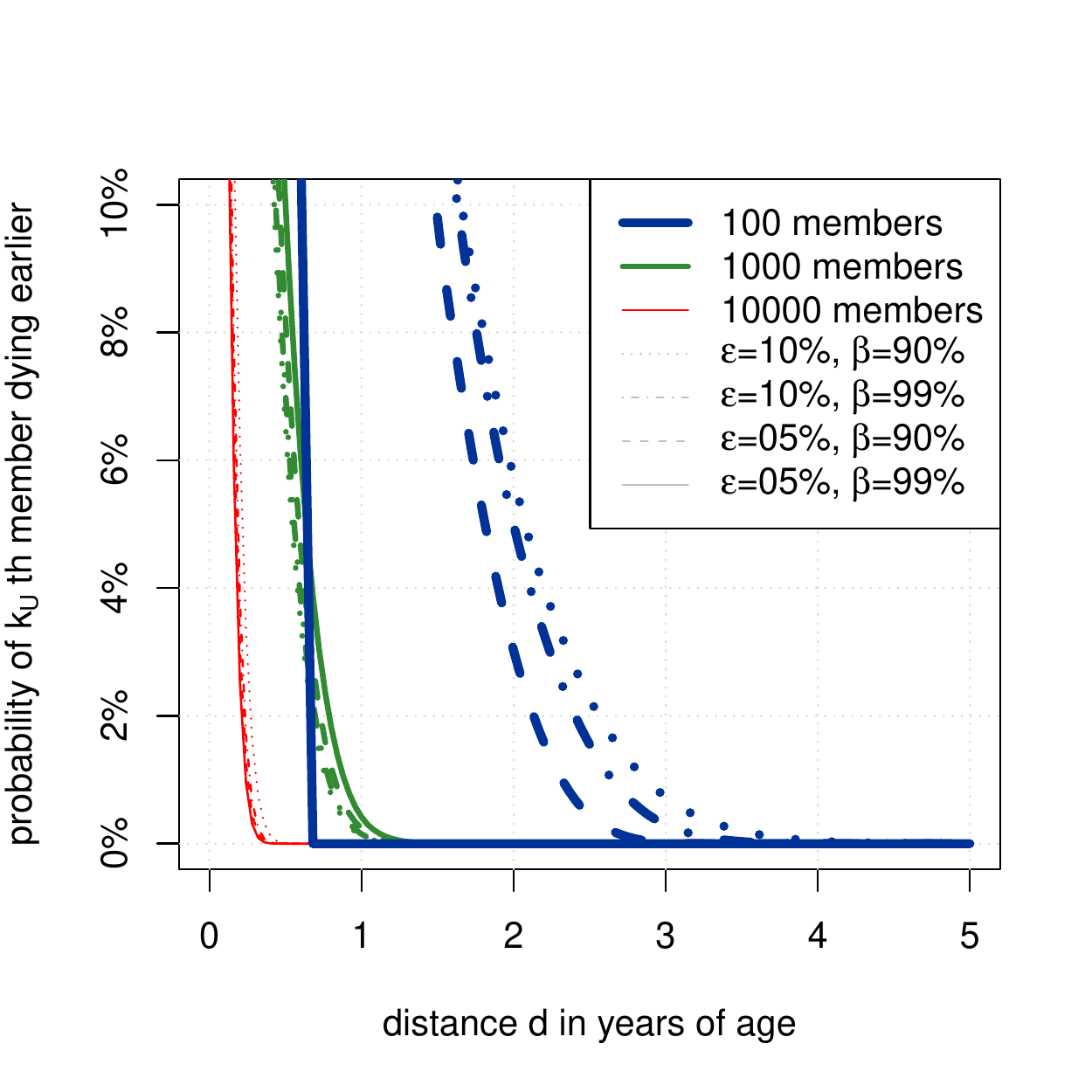}
    \caption{$\mathbb{P}[T_{(k_{U})} \leq F^{-1}(k_{U}/N)-d]$ plotted against $d$, in which $F(s)={}_{s}q_{70}$ is calculated using the UK-based life table S1PFL \citep{IFoA2008}.  The plots are shown for a selection of values of the initial number $N$ of 70-year-old members in the fund, certainties $\beta$ and threshold parameter $\varepsilon$ and when a symmetric lower and upper income threshold are imposed.  The plot using the Human Mortality Database's life table for the UK for 2016 \citep{HMD2016} looks very similar and is omitted.}
		\label{fig:beta-Time-lower-upper}
\end{figure}

\begin{figure}[H]
    \centering
\includegraphics[trim=0 0.5cm 0 1.5cm,clip,width=0.6\textwidth]{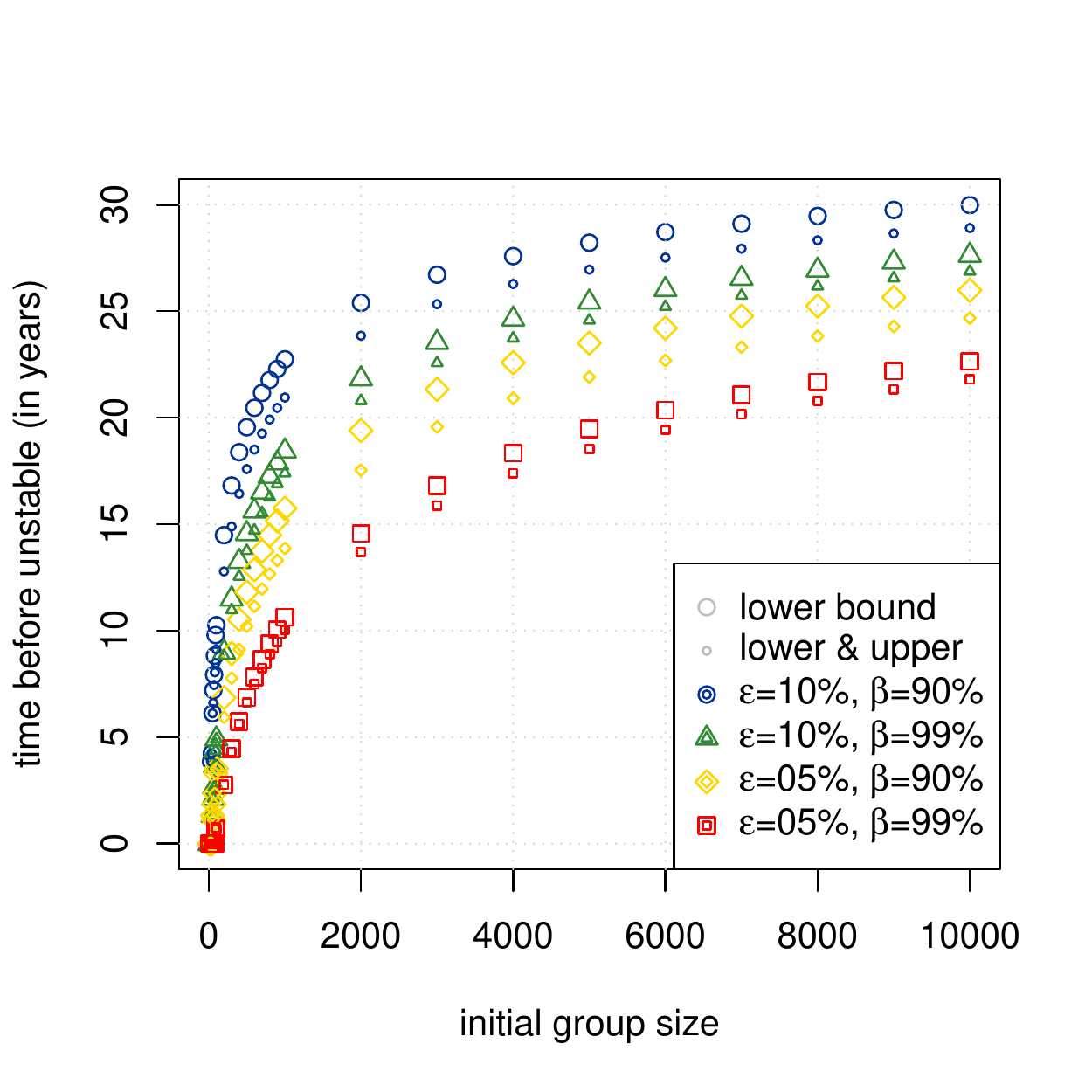}
    \caption{The likely length of time for which the fund can provide a stable income for a group of $70$-year-olds, plotted as a function of the initial number of members in the fund $N$.  The plot is based on the UK-based life table S1PFL \citep{IFoA2008}.  The smaller symbols indicate a symmetric lower and upper threshold on the income, and the larger symbols indicate a lower threshold only on the income.}
		\label{fig:Time-IFoA}
\end{figure}

Our results also provide evidence to support broadly the statement in \citet[page 967]{qiaosherris2013} that ``Most of the significant pooling benefits are realized when the pool size reaches $1\,000$''; Figure \ref{fig:Time-IFoA} indicates a slowing up of the rate of increase of the likely time for which a stable income can be provided around $N=2\,000$ (the study of \citealt{qiaosherris2013} considered only $N\in\{1,100,1\,000, 10\,000\}$).  It is important to note that when the number of members who receive a stable income is considered, there is still a steadily increasing proportion of members who receive a life-long, stable income, as can be observed from Figure \ref{fig:k-graph}.

\section{Approximate formula}\label{section:Donsker-approach}
Here an approximate formula to determine the maximum integer $k:=k_{U}$ that satisfies inequality (\ref{eq:problowerbound}) is presented.  Fix the initial number of fund members $N$, certainty $\beta \in (0,1)$, threshold parameter $\varepsilon \in (0,1)$ and let $\Phi^{-1}$ be the inverse of the standard normal distribution function $\Phi$.  Then
\begin{align}\label{eq:approxII}
 k_{U} \approx  k_{U}^{\textrm{approx}} := N \left( 1 - \left\lfloor \frac{1}{1-\varepsilon} \left(1-\frac{1}{1+\tfrac{1}{N} \left( \tfrac{1-\varepsilon}{\varepsilon} \right)^2\left( \Phi^{-1} \left(\tfrac{1-\beta}{2}\right) \right)^2} \right) \right\rfloor_{N} \right),
\end{align}
in which $\lfloor u \rfloor_N=\max\{i/N : i/N \leq u \textrm{ and } i \in \{0,1,\ldots,N\}\}$ for $u \geq 0$.  The derivation is shown in Appendix \ref{APPproofs}.

Figure \ref{fig:error-absolute} shows that $1-k_{U}^{\textrm{approx}}/N$ is visually almost indistinguishable from $1-k_{U}/N$ on a scale from $0\%$ to $100\%$.  However, the relative errors reveal small systematic discrepancies.  Observe in Figure \ref{fig:error-relative} that the relative errors change most noticeably with the threshold parameter $\varepsilon$, upon fixing the certainty $\beta$.  On the other hand, the relative errors are similar for different certainties $\beta$, upon fixing the threshold parameter $\varepsilon$.  This indicates that the treatment of $\varepsilon$ is the major source of error.   In fact, in the derivation of the approximation (\ref{eq:approxII}), a vague argument that involves $\varepsilon$ and works by compensating one effect with another is employed; see the expressions (\ref{eq:$1-beta=P[<supu<Fu-hatF_nF^-1]$}) and (\ref{eq:$1-beta=P[<supBu]$}).

\begin{figure}[H]
\captionsetup[subfigure]{justification=centering} %
    \centering
    \begin{subfigure}[b]{0.5\textwidth}
        \centering
\includegraphics[trim=0 0.5cm 0 1.5cm,clip,width=\textwidth]{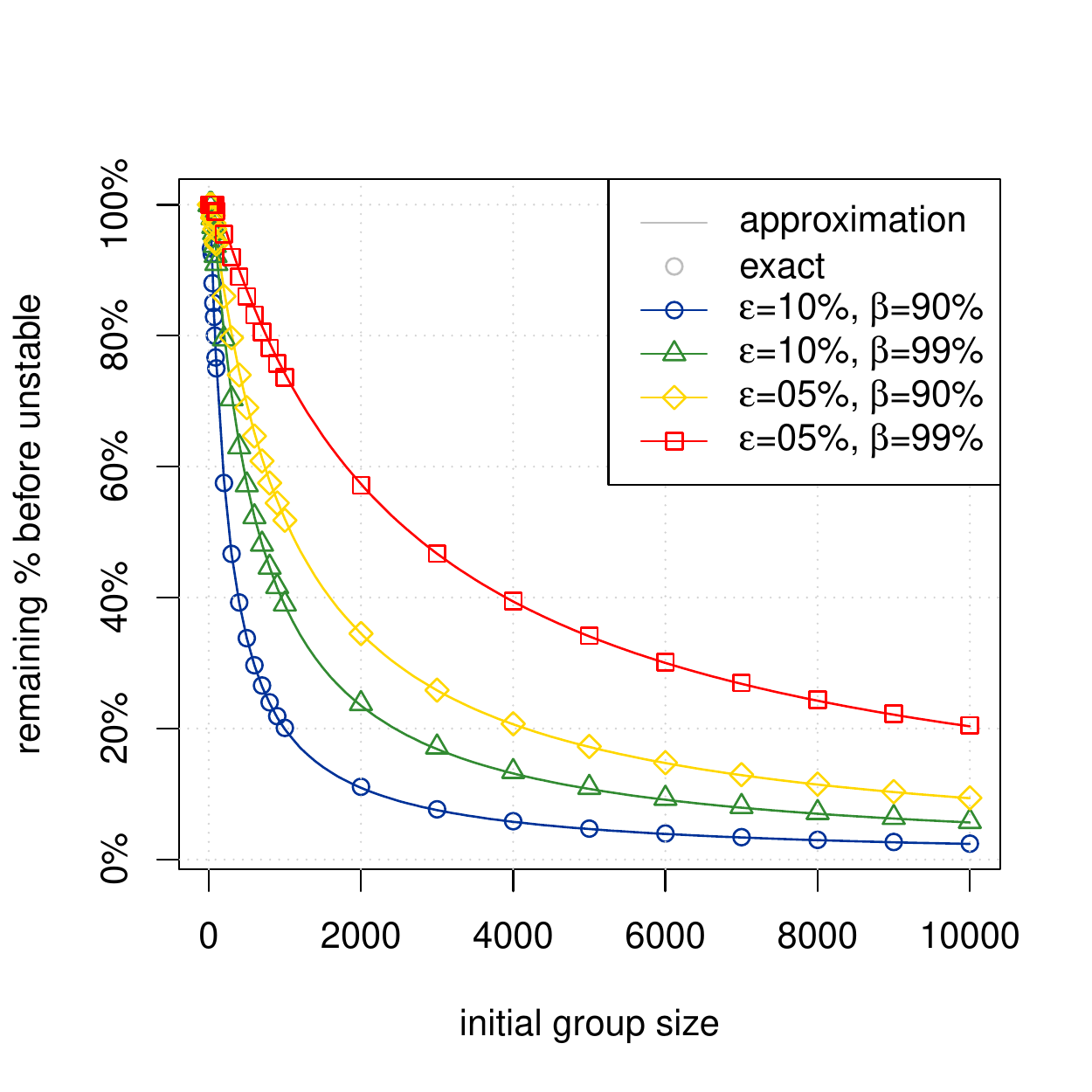}
\caption{$100(1-k_{U}/N)$ versus $100(1-k_{U}^{\textrm{approx}}/N)$.}
\label{fig:error-absolute}
    \end{subfigure}%
    \begin{subfigure}[b]{0.5\textwidth}
        \centering
\includegraphics[trim=0 0.5cm 0 1.5cm,clip,width=\textwidth]{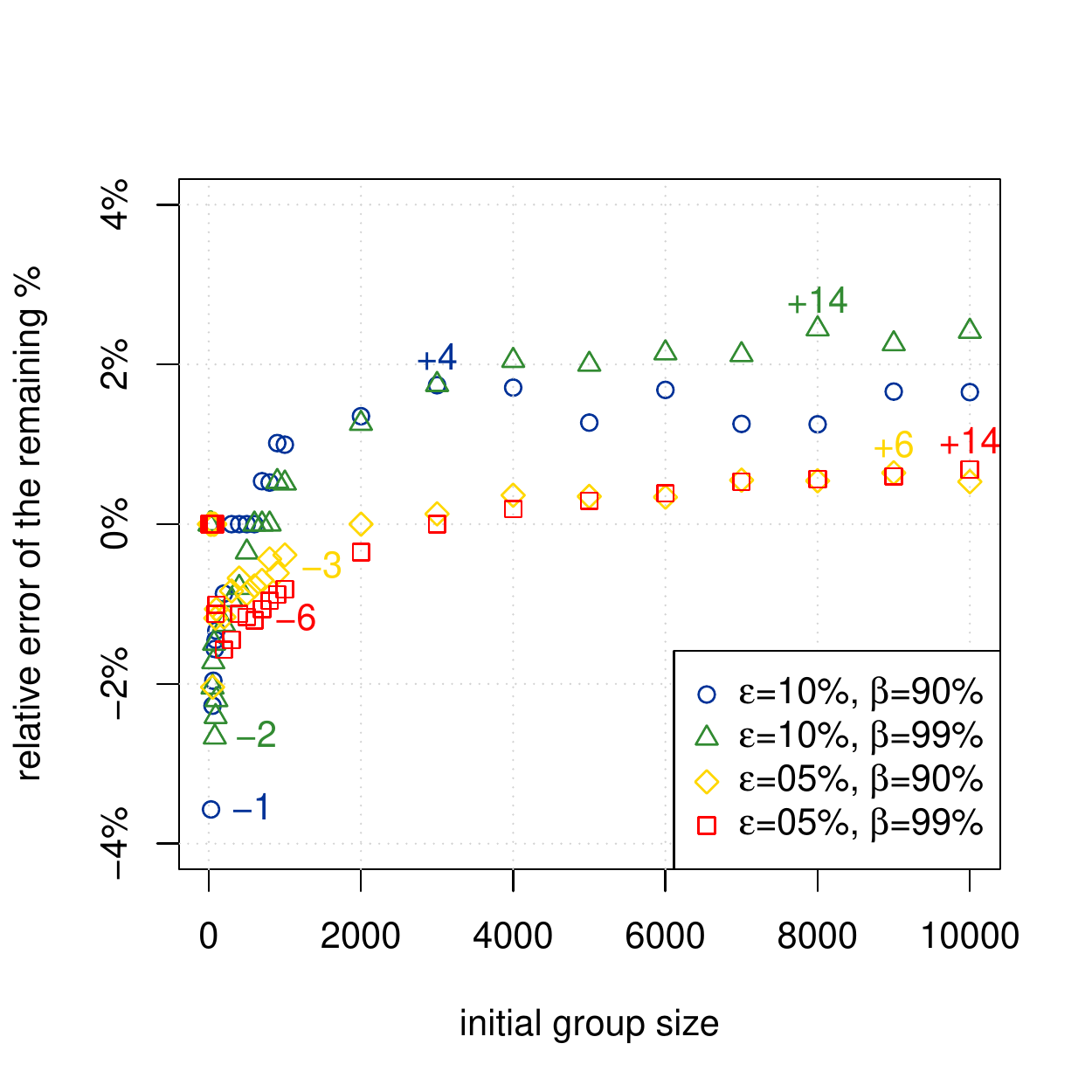}
\caption{$(k_{U}^{\textrm{approx}}-k_{U})/(N-k_{U})$.}
\label{fig:error-relative}
    \end{subfigure}
    \caption{Representations of the error in the approximation $k_{U}^{\textrm{approx}}$.  Figure \ref{fig:error-absolute} plots the percentage of members still alive after the $k_{U}$th death, i.e. $100(1-k_{U}/N)$ and displayed as discrete symbols, and its approximation $100(1-k_{U}^{\textrm{approx}}/N)$, plotted as solid lines, as the initial number of members $N$ increases.  The results are shown for various choices of the certainty $\beta$ and threshold parameter $\varepsilon$.  Figure \ref{fig:error-relative} plots the relative error of these values, which reduces to calculating $(k_{U}^{\textrm{approx}}-k_{U})/(N-k_{U})$.  The displayed integers inside the graph show the maximal value of $k_{U}^{\textrm{approx}}-k_{U}$ in terms of number of members.}
		\label{fig:errorinapproximation}
\end{figure}

More generally, it is possible to apply the same methodology to find an approximate formula to determine the maximum integer $k:=k_U$ that satisfies
\[
\mathbb{P}\left[(1-\varepsilon_1)\tfrac{i-1}{N}+\varepsilon_1\geq U_{(i)}\geq(1+\varepsilon_2)\tfrac{\min\{i,N-1\}}{N}-\varepsilon_2\,\mbox{ for
all $i\in\{1,2,\ldots,k\}$}\right]\geq\beta.
\]
However, the result involves a special function that is as inexplicit as the target function, and we give it only for completeness:
\[
k_U \approx k_U^{\textrm{approx}}:=N \, \lfloor \Psi^{-1}_{\varepsilon_1,\varepsilon_2,N}(\beta)\rfloor_N
\]
with
\[
\Psi_{\varepsilon_1,\varepsilon_2,N}(y) :=\mathbb{P}\left[-\sqrt{N}\tfrac{\varepsilon_2}{1+\varepsilon_2}\leq\inf_{s\leq\frac{1}{(1+\varepsilon_2)
(1-y)}-1}W(s);\sup_{s\leq\frac{1}{(1-\varepsilon_1)(1-y)}-1}W(s)\leq\sqrt{N}\tfrac{\varepsilon_1}{1-\varepsilon_1}\right]
\]
for all $y \in (0,1)$ and $W$ a standard Brownian motion.

\section{Conclusion}
Here a pooled annuity fund has been analyzed to see how well it provides a stable income to its participants, with a specified degree of certainty.  Income stability is defined path-wise.  The main theoretical result is the derivation of a lower bound that can be used as a good approximation to the true number of members.   The lower bound gives both insights and a greater understanding of the pooled annuity fund's ability to diversify idiosyncratic longevity risk:
\begin{itemize}
\item It is independent of the mortality distribution of the fund members.  As the lower bound is close to the true value, it suggests that the true number of members is unlikely to change much if the mortality distribution is varied, all other parameters being unchanged.
\item It can be calculated faster than the true number, since the former requires only the simulation of the order statistics of independent, standard uniform random variables rather than a time-dependent stochastic simulation.
\item When a stable income is defined as the income lying between a symmetric upper and non-zero lower income threshold, the longest lived members do not receive a stable income in any future state of the world.  Even if such members may receive a stable income for the majority of their lifetime, this means that there should be an alternative method of providing a stable income with a desired degree of certainty for the entirety of their lifetime, for example by using deferred life annuity contracts or the opportunity to exit the fund at a certain age or the modifications proposed in, for example, \citet{chenetal2019}, \citet{chenrach2019} and \citet{DonnellyYoung2017}.  This is an insight into the failure of the diversification of idiosyncratic longevity risk for the longest lived members.
\end{itemize}

The main theorem quantifies the number of members required to give the desired degree of income stability.  In other words, the number of members needed to diversify idiosyncratic longevity risk to a specific degree.   At a high-level, the results suggest that the membership of the fund should number in the thousands if it is important to reduce the income instability derived from longevity risk pooling to low levels.  Our finding complements the results of \citet{qiaosherris2013}, who suggest the same order of magnitude based on examining the quantiles of the income.  However, our results allow the calculation of the precise number of members required and, moreover, use a path-wise definition of income stability.  This means that we know how the income streams have behaved up to the time that the calculated number of people die, rather than only how they behave at one point in time. 

Applying the theoretical results, the length of time for which the pooled annuity fund could provide a stable income with a specified degree of certainty was determined.  Again, the length of time depends on how income stability is specified.  It also depends on a mortality law; the results for one law are presented.  Based on the chosen mortality law, the fund can provide a stable income from $10$ years to $30$ years depending on the specification of income stability and the initial number of members in the fund.  Finally, an approximation to the lower bound is derived, and is found to be a close approximation to it.

The limitations of the assumptions made in this paper motivate future research into this area.  It is assumed that all annuitants join with the same amount of money and are the same age.  While it would be possible to impose such restrictions in real life, it is overly prescriptive and makes the continued success of the pooled annuity fund less likely.  All annuitants are assumed to have the same chance of dying at each age.  How well this approximates reality will depend on the group of annuitants; it would be unrealistic if anyone was allowed to join but may be appropriate for a fund offered only to university professors, for example.  Deaths are assumed to occur independently, but this may not be true if annuitants catch a deadly virus from each other.  These restrictions on the characteristics of the fund membership could be relaxed in future research.

The historical record of actuaries under-estimating future lifetimes and the fast spread across the world of the coronavirus both suggest that the distribution of the annuitants' future lifetimes assumed in the model is unlikely to be observed in practice.  How the inclusion of systematic longevity risk affects the results very broadly is discussed briefly here.  Suppose that the day after the fund started, the future lifetime distribution of the annuitants was found to understate for how long the annuitants would live.  Assume for simplicity that the definition of stability includes only a lower income threshold.

While the income paid to the annuitants would be consequently revised downwards, the number $k_{U}$ who would receive the stable version of this new income for life would be unchanged (assuming no further updates to the future lifetime distribution are needed).  This is because $k_{U}$ is independent of the choice of the future lifetime distribution, but it requires that this distribution is the one observed in practice.  This would be the case if the distribution changed the day after the fund started, and never changed subsequently.  

However, suppose that the understatement of future lifetime was not discovered while the fund's annuitants were alive.  Then the income paid to the annuitants would be too high and would gradually decline, as fewer people died than expected meaning that too much money was paid out.  Thus fewer than $k_{U}$ people would receive a stable, lifelong income for some given level of certainty since the declining income would breach the lower income threshold earlier than anticipated.

If instead more people died than suggested by the future lifetime distribution used to calculate the annuity payments, then less income is paid out than expected.  In that case, more than $k_{U}$ people would receive a stable, lifelong income, because more people die earlier - thus getting a lifelong income over a shorter lifetime - and there is more money for the longer-lived.

As the focus of this paper is idiosyncratic longevity risk, the quantification of the effect of systematic longevity risk on the results requires a further study.  This is important so that how the annuitants' income may vary over their lifetime is understood better, thus allowing the pooled annuity fund to meet the needs and expectations of annuitants.

\section*{Acknowledgments}
Research for this paper was undertaken as part of the Institute and Faculty of Actuaries' ARC research programme ``Minimising Longevity and Investment Risk while Optimising Future Pension Plans'', for which funding is gratefully acknowledged by the authors.  The authors thank three anonymous reviewers and the Editor, Prof. Dr W\"{u}thrich, for criticisms and suggestions which have undoubtedly improved the paper immensely.



\newpage
\appendix
\section*{Appendix}

\section{Derivation} \label{APPproofs}

\begin{claim}
For the order statistics $(U_{(i)})_{i=1}^N$ of $N$ independent and standard uniformly distributed random variables, certainty $\beta\in[0,1]$, and threshold parameter $\varepsilon\in(0,1)$, let $k_U\leq N$ be the last integer that fulfills
\[
\mathbb{P} \left[(1-\varepsilon)\tfrac{i-1}{N} + \varepsilon \geq U_{(i)}\;\;\mbox{for all $i\in\{1,2,\ldots, k\}$}\right] \geq \beta
\]
For $u\geq0$, let $\lfloor u \rfloor_N=\max\{i/N: i/N \leq u \, \textrm{ and } \, i \in \{0,1,\ldots,N]\}\}$, and let $\Phi^{-1}$ be the inverse of the standard normal distribution function $\Phi$.  Then $k_U$ can be calculated approximately from
\[
    1-\frac{k_U}{N}\approx \left\lfloor \frac{1}{1-\varepsilon} \left(1-\frac{1}{1+\tfrac{1}{N} \left( \tfrac{1-\varepsilon}{\varepsilon} \right)^2\left( \Phi^{-1} \left(\tfrac{1-\beta}{2}\right) \right)^2} \right) \right\rfloor_N.
\]
\end{claim}
\begin{proof}[Derivation of claim]
Assume that both the time of deaths of members and payment times to surviving members are reasonably dense in time so that from (\ref{EQNfirstthmiii}) in the proof of Theorem \ref{theorem:main},
\[
\mathbb{P}\left[ \inf_{s \in [0,T_{(k_U)}) } \frac{{}_{s}p_{x}}{{}_{s}\hat{p}_{x}} \geq 1-\varepsilon \right]\geq\beta 
\]
is approximately equivalent to
\[
\mathbb{P}\left[ \inf_{s \in [0,t] } \frac{{}_{s}p_{x}}{{}_{s}\hat{p}_{x}} \geq 1-\varepsilon \right]\geq\beta,\quad\mbox{with $t\approx T_{(k)}$}.
\]
Using $\hat{F}_{N}(s)= 1 - {}_{s}\hat{p}_{x}$ and $F(s) = 1 - {}_{s}p_{x}$ from the proof of Theorem \ref{theorem:main}, 
\[
\begin{split}
\beta = \mathbb{P}\left[\frac{\px[s]{x}}{\actsymb[s]{\hat{p}}{x}}\geq 1-\varepsilon,\;\;\forall\,s\leq t\right] = & \mathbb{P}\left[\frac{1-F(s)}{1-\hat{F}_{N} (s)}\geq 1-\varepsilon,\;\;\forall\,s\leq t\right] \\
= & \mathbb{P}\left[\frac{\varepsilon}{1-\varepsilon}\geq\frac{F(s)-\hat{F}_{N} (s)}{1-F(s)},\;\;\forall\,s\leq t\right].
\end{split}
\]
It is more convenient to look at the complementary probability, i.e.
\[
1-\beta = \mathbb{P}\left[\exists\,s\leq t:\;\frac{\varepsilon}{1-\varepsilon}<\frac{F(s)-\hat{F}_{N} (s)}{1-F(s)}\right] = \mathbb{P}\left[\frac{\varepsilon}{1-\varepsilon}<\sup_{s\leq t}\frac{F(s)-\hat{F}_{N} (s)}{1-F(s)}\right].
\]
Let $F^{-1}$ be the generalized inverse of $F$.  Re-writing the above expression in a form amenable to applying well-known results from the literature yields
\begin{equation}\label{eq:$1-beta=P[<supu<Fu-hatF_nF^-1]$}
    1-\beta=\mathbb{P}\left[\frac{\varepsilon}{1-\varepsilon}<\sup_{u\leq F(t)}\frac{u-\hat{F}_{N}(F^{-1}(u))}{1-u}\right].
\end{equation}
One of the early versions of Donsker's theorem states that the process $\,u\mapsto\sqrt{N}(u-\hat{F}_N(F^{-1}(u)))\,$ converges in distribution for $N\rightarrow\infty$ to a Brownian bridge \citep{donsker1952}. However, applying this theorem to (\ref{eq:$1-beta=P[<supu<Fu-hatF_nF^-1]$}) yields an approximation that noticeably underestimates the true value. The issue seems to be that the discontinuous jump-processes $(\hat{F}_N)_{N=1}^\infty$ are approximated by a continuous Brownian bridge.  Heuristically, a continuous process should be a better fit for a continuous process than a discontinuous one. In fact, the classical Donsker theorem is based on continuous interpolations of random variables \citep{Billingsley1999}.

Let $\hat{F}^c_n$ be the process that continuously interpolates the jumps of $\hat{F}_n$ for given $n$.  Then $\hat{F}^c_n\geq\hat{F}_n$ are functions on the real line. This leads to a smaller supremum and, in turn, to an under-estimation of (\ref{eq:$1-beta=P[<supu<Fu-hatF_nF^-1]$}) when $\hat{F}_n$ is replaced with its continuous version $\hat{F}^c_n$ in (\ref{eq:$1-beta=P[<supu<Fu-hatF_nF^-1]$}).

One way to compensate for this under-estimation when using the Brownian bridge is to enlarge the set $\{u:u\leq F(t)\}$. A similar argument has been used in \citet{Vrbik2018} for small sample corrections in the context of the Kolmogorov-Smirnov test. Interestingly, there is an enlargement that leaves the value of (\ref{eq:$1-beta=P[<supu<Fu-hatF_nF^-1]$}) almost unchanged: let $y$ be the value of the last jump of $\hat{F}_{N}(F^{-1}(u))$ in $\{u:u\leq F(t)\}$. Assume that $\varepsilon/(1-\varepsilon)$ hasn't been crossed so far (the other case doesn't change the outcome of the inequality in (\ref{eq:$1-beta=P[<supu<Fu-hatF_nF^-1]$}) for any enlargement).  Then
\begin{align*}
    \frac{\varepsilon}{1-\varepsilon}\geq\frac{u-y}{1-u}\quad\mbox{for $\,u\leq F(t)$}.
 \end{align*} 
 Rearranging the inequality and using $\,y\approx F(t)$ gives
\[
\varepsilon+(1-\varepsilon)y\geq u,\qquad \qquad \mbox{or} \qquad \qquad \varepsilon+(1-\varepsilon)F(t)\geq u.
\]
Hence, $\{u:u\leq F(t)\}$ can be approximately enlarged to $\{u:u\leq\varepsilon+(1-\varepsilon)F(t)\}$ without changing the value in (\ref{eq:$1-beta=P[<supu<Fu-hatF_nF^-1]$}).

Overall, let $B$ be a Brownian bridge in $[0,1]$, then
\begin{align}\label{eq:$1-beta=P[<supBu]$}
    1-\beta&\approx\mathbb{P}\left[\frac{\varepsilon}{1-\varepsilon}<\sup_{u\leq\varepsilon+(1-\varepsilon)F(t)}\frac{B(u)}{\sqrt{N}(1-u)}\right].
\end{align}
\indent
The process $\,u\mapsto B(u)/(1-u)\,$ is a time-changed Brownian motion, as it is a continuous, zero-mean, Gaussian process with covariance function $\,(u,v)\mapsto u/(1-u)\,$ for $\,0\leq u\leq v\leq1$ \citep[pp.103-104]{Karatzas1988}.

Then, letting $W$ be a Brownian motion,  
\begin{align*}
    1-\beta \approx\mathbb{P}\left[\frac{\varepsilon}{1-\varepsilon}<\sup_{u\leq\varepsilon+(1-\varepsilon)F(t)}\frac{W(\textstyle\frac{u}{1-u})}{\sqrt{N}}\right] &= 2\,\mathbb{P}\left[\sqrt{N}\frac{\varepsilon}{1-\varepsilon}<W(\textstyle\frac{\varepsilon+(1-\varepsilon)F(t)}{1-(\varepsilon+(1-\varepsilon)F(t))})\right]
    \\&=2\,\Phi\left(-\sqrt{\tfrac{1-(\varepsilon+(1-\varepsilon)F(t))}{\varepsilon+(1-\varepsilon)F(t)}}\sqrt{N}\frac{\varepsilon}{1-\varepsilon}\right)
    \\&=2\,\Phi\left(-\sqrt{\tfrac{1}{1-(1-\varepsilon)\px[t]{x}}-1}\sqrt{N}\frac{\varepsilon}{1-\varepsilon}\right).
\end{align*}
Rearranging yields
\begin{align*}
    \px[t]{x}\approx\frac{1}{1-\varepsilon}\left(1-\frac{1}{1+\tfrac{1}{N}\left(\tfrac{1-\varepsilon}{\varepsilon}\right)^2 \left(\Phi^{-1} \left(\tfrac{1-\beta}{2} \right) \right)^2}\right).
\end{align*}
Observe that $\,\px[t]{x}\approx\actsymb[t]{\hat{p}}{x}\approx\actsymb[T_{(k_U)}]{\hat{p}}{x}=L_{x+T_{(k_U)}}/N=(N-k_U)/N=1-k_U/N$. 
Hence, the claim follows from $\px[t]{x}\approx1-k_U/N$ together with $k_U\leq N$ is an integer.
\end{proof}

\end{document}